\documentclass{imsart}
\RequirePackage[OT1]{fontenc}
\RequirePackage{amsthm,amsmath,natbib}
\RequirePackage[colorlinks,citecolor=blue,urlcolor=blue]{hyperref}

\usepackage{geometry}
\RequirePackage{amssymb,mathtools }
\usepackage{booktabs}
\usepackage[table]{xcolor}
\usepackage{caption}
\usepackage{algorithm,algpseudocode,float}
\usepackage{eurosym}
\usepackage{tikz}
\usetikzlibrary{fit,positioning,arrows,calc}
\tikzset{
  main/.style={circle, minimum size = 5mm, thick, draw =black!80, node distance = 10mm},
  connect/.style={-latex, thick},
  box/.style={rectangle, draw=black!100},
  edge/.style={->,> = latex'}
}
\setcitestyle{authoryear,open={(},close={)}}

\startlocaldefs
\numberwithin{equation}{section}
\theoremstyle{plain}

\theoremstyle{remark}
\newtheorem{rem}{Remark}[section]
\newtheorem{lemma}{Lemma}[section]
\newtheorem{proposition}{Proposition}
 
\endlocaldefs

\begin{document}

\begin{frontmatter}
\title{Kalman filter demystified: from intuition to probabilistic graphical model to real case in financial markets}
\runtitle{Kalman filter demystified}

\begin{aug}
\author{\fnms{Eric} \snm{Benhamou}\thanksref{a,b,e1}}

\address[a]{A.I. SQUARE CONNECT, \\ 35 Boulevard d'Inkermann 92200 Neuilly sur Seine, France}
\address[b]{Affiliated researcher to LAMSADE (UMR CNRS 7243)\\ and QMI (Quantitative Management Initiative) chair,\\ Université Paris Dauphine, \\ Place du Maréchal de Lattre de Tassigny,75016 Paris, France}
\thankstext{e1}{eric.benhamou@aisquareconnect.com, eric.benhamou@dauphine.eu}

\runauthor{E. Benhamou}
\affiliation{A.I. SQUARE CONNECT, LAMSADE and QMI, Paris Dauphine University}
\end{aug}

\begin{abstract}
In this paper, we revisit the Kalman filter theory. After giving the intuition on a simplified financial markets example, we revisit the maths underlying it. We then show that Kalman filter can be presented in a very different fashion using graphical models. This enables us to establish the connection between Kalman filter and Hidden Markov Models. We then look at their application in financial markets and provide various intuitions in terms of their applicability for complex systems such as financial markets. Although this paper has been written more like a self contained work connecting Kalman filter to Hidden Markov Models and hence revisiting well known and establish results, it contains new results and brings additional contributions to the field. First, leveraging on the link between Kalman filter and HMM, it gives new algorithms for inference for extended Kalman filters. Second, it presents an alternative to the traditional estimation of parameters using EM algorithm thanks to the usage of CMA-ES optimization. Third, it examines the application of Kalman filter and its Hidden Markov models version to financial markets, providing various dynamics assumptions and tests. We conclude by connecting Kalman filter approach to trend following technical analysis system and showing their superior performances for trend following detection.
\end{abstract}


\end{frontmatter}


\section{Introduction}
One of the most challenging question in finance is to be able from past observation to make some meaningful forecast. Obviously, as no one is God, no one would ultimately be able to predict with one hundred percent accuracy where for instance the price of a stock will end. However, between a pure blind guess and a perfect accuracy forecast, there is room for improvement and reasoning. If in addition, we are able somehow to model the dynamics of the stock price and factor in some noise due to unpredictable human behavior, we can leverage this model information to make an informed estimate, in a sense an educated guess that cannot miss the ballpark figure. It will not be 100 percent accurate but it is much better than a pure random appraisal. Indeed, this scientific question of using a model and filtering noise has been extensively examined in various fields: control theory leading to Kalman filter, Markov processes leading to hidden Markov models and lately machine learning using Bayesian probabilistic graphical models. In this work, we revisit these three fields to give a didactic presentation of this three approaches and emphasizing the profound connection between Kalman filter and Hidden Markov Models (HMM) thanks to Bayesian Probabilistic Graphical models. We show in particular that the derivation of Kalman filter equations is much easier once we use the more general framework of Bayesian Probabilistic Graphical models. In particular, we show that Kalman filter equations are just a rewriting of the sum product algorithm (also referred to as the Viterbi algorithm for HMM). We then provide various dynamics assumptions for financial markets and comment their overall performance. We compare these trading system with simpler moving average trend detection trading systems and show that they provide better performances.

\section{Intuition}\label{TheIntuition}
In a nutshell, a Kalman filter is a method for predicting the future state of a system based on previous ones. \\
It was discovered in the early 1960's when Kalman introduced the method as a different approach to statistical prediction and filtering (see \cite{Kalman_1960} and \cite{Kalman_1961}). The idea is to estimate the state of a noisy system. Historically, it was build to monitor the position and the speed of an orbiting vehicle. However, this can be applied to non physical system like economic system. In this particular settings, the state will be described by estimated equations rather than exact physical laws. We will emphasize this point in the section \ref{InPractice} dealing with practical application in finance.

To build an intuition, let us walk through a simple example - if you are given the data with green dots (that represent, by the way, the predicted Kalman filtered predicted price of the S\&P500 index on October 17, 2018), it seems reasonable to predict that the orange dot should follow, by simply extrapolating the trend from previous samples and inferring some periodicity of the signal. However, how assured would you be anticipating the dark red point on the right (that is just 10 minutes later)? Moreover, how certain would you be about predicting the orange point, if you were given the black series (that represent the real prices) instead of the green one?\\

From this simple example, we can learn three important rules:
\begin{itemize}
\item Predicting far ahead in the future is less reliable than near ahead.
\item The reliability of your data (the noise) influences the reliability of your forecast. 
\item It's not good enough to give a prediction - you also want to provide a confidence interval.
\end{itemize}

\begin{center}
\includegraphics[width=10cm]{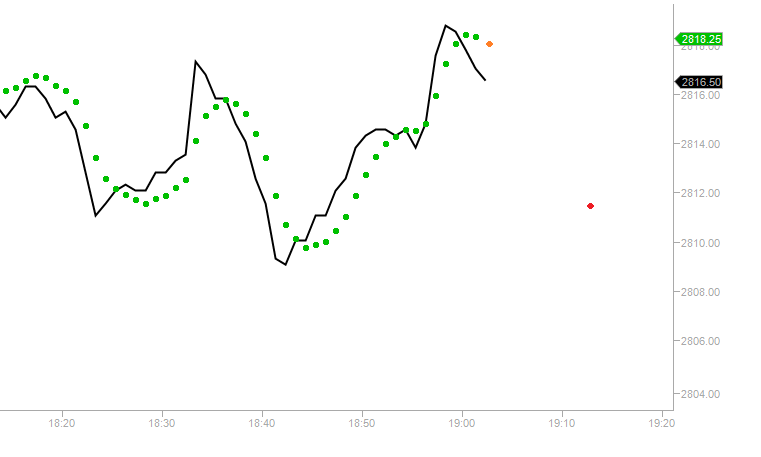}
\captionof{figure}{A simple example from financial markets to build an intuition of Kalman filter}
\label{fig1}
\end{center}

From this simple presentation, we can conceptualize the approach and explain what a Kalman filter is.
First, we need a \textbf{state}. A state is a representation of all the parameters needed to describe the current system and perform its prediction. In our previous example, we can define the state by two numbers: the price of the S\&P 500 at time $t$, denoted by $p(t)$ and the slope of the price at time $t$, $s(t)$. In general, the state is a \textbf{vector}, commonly denoted $\bf{x}$. Of course, you can include many more parameters if you wish to accommodate more complex systems.

Second, we need a \textbf{model}. The model describes how the system behaves. It provides the underlying equations that rules our system. It may be an ideal representation or a simplified version of our system. It may be justified by some physical laws (in the case for instance of a Kalman filter for a GPS system), or by some empirical analysis (in the case of finance). In the standard Kalman filter, the model is always a linear function of the state. In extended Kalman filter, it is a non linear function of the state. In our previous example, our model could be:
\begin{itemize}
\item the price at time $t$, $p(t)$, is obtained as the previous price $p(t-1)$ plus the slope at time $t-1$: $s(t-1)$
\begin{equation}
p(t) = p(t-1) + s(t-1)
\end{equation}

\item the slope is evolving over time with a periodic sinusoidal function $\psi(t) = a \sin(b t + c )$
\begin{equation}
s(t) = s(t-1) + \psi(t) 
\end{equation}
\end{itemize}

Usually, one expresses this model in matrix form to make it simple. We have
\begin{eqnarray}
 \underbrace{\left(\begin{array}{c} p(t) \\ s(t)\\ \end{array} \right) }_{{\bf{x}}_{t}} & = &
 \underbrace{\left(\begin{array}{cc} 1  & 1\\ 0 &  1 \\ \end{array} \right)}_{\bf{F}} \cdot
 \underbrace{\left(\begin{array}{c} p(t-1)\\ s(t-1)\\ \end{array} \right)}_{{\bf{x}}_{t-1}} + \underbrace{
\left(\begin{array}{cc} 0 & 0 \\ 0 & 1 \end{array} \right) \cdot \left(\begin{array}{c} 0 \\ \psi(t) \end{array} \right)}_{\bf{B}_t \cdot {\bf{u}}_{t} } \label{state_equation_1}
\end{eqnarray}

In the above equation, following standard practice, we have denoted by ${{\bf{x}}_{t}}= \left( p(t), s(t) \right)^T$, the \textbf{state vector} that combines the price $p(t)$ and the slope $s(t)$. We have denoted by ${\bf{u}}_{t} =\left( 0, \psi(t) \right)^T$, the \textbf{control vector}. The idea underlying the control vector is that it can control somehow the dynamic of the state vector. The previous equation (\ref{state_equation_1}) is called the \textbf{state equation} and writes 
\begin{equation}
{\bf{x}}_{t} = {\bf{F}}  \cdot {\bf{x}}_{t-1} + {\bf{B}}_t \cdot {\bf{u}}_{t}
\end{equation}
Of course, our model is too simple, else we wouldn't need a Kalman Filter! To make it more realistic, we add an additional term - the \textbf{noise process}, $\bf{w_t}$ that represents anything else that we do not measure in our model. Although we don't know the actual value of the noise, we assume we can estimate how "noisy" the noise is and can make assumptions about the noise distribution. In standard Kalman modeling, we keep things simple and assumed this noise to be normally distributed. To measure how "noisy" the noise is, we refer to the variance of the normal distribution underpinning the noise $\bf{w_{t-1}}$. The \textbf{state equation} is modified into:

\begin{equation}\label{state_equation}
{\bf{x}}_{t} = {\bf{F}} {\bf{x}}_{t-1} + {\bf{B}}_t \cdot {\bf{u}}_{t} + {\bf{w}}_{t-1}
\end{equation}

Third, we need to use \textbf{measurement} to improve our model. When new data arrived, we would like to change the value of our model parameters to reflect our improved understanding of the state dynamic. We would proceed in two steps: make a prediction and then a correction based on what has happened in time $t$. There is  a subtleties here. What we measure does not have to be exactly the states. It has to be related but not the same. For instance, in a Kalman filter system in a GPS, the state could be the GPS's acceleration, speed and position, while what we measure may be the car's position and the wheel velocity. In our financial example, we may only be able to observe the price but not the slope. Our measure would therefore be limited to the price in this particular setting. Hence, the measurement would be represented as follows:

\begin{equation}
\text{measurement} = 
\left(\begin{array}{cc} 0 & 1 \end{array} \right) \cdot
\left(\begin{array}{c} p(t) \\ s(t) \end{array} \right) 
\end{equation}

In general, the measurement should be a vector, denoted by $\bf{z}$, as we may have more than one number for our measurement. Measurements should also be noisy to reflect that we do not measure perfectly. Hence, the equation specify measurement would write:
\begin{equation}
{\bf{z}}_t = {\bf{H}} {\bf{x}}_t +{\bf{v}}_t
\end{equation}
Where ${\bf{v}}_t$ is the \textbf{measurement noise}, and $H$ is a matrix with rows equal to the number of measurement variables and columns equal to the number of state variables.

Fourth, we need to work on \textbf{prediction} as we have set up the scene in terms of modeling. This is the key part in Kalman filter.
To build intuition, let us assume, for the time being, noise is equal to zero. 
This implies our model is perfect. How would you predict the state at time $t$, ${\bf{x}}_{t}$ knowing the previous state (at time $t-1$)? 
This is dammed simple. 
Just use the state equation and compute our prediction of the state, denoted by ${\bf{\hat x}}_{t }$ as:
\begin{equation}
{\bf{\hat x}}_{t}= {\bf{F}} {\bf{x}}_{t-1} + {\bf{B}}_t  {\bf{u}}_{t-1}
\end{equation}

Hold on! In the mean time, we do some measurements reflected by our measurement equation:
\begin{equation}
{\bf{\hat z}}_{t} = {\bf{H}} {\bf{\hat x}}_{t}
\end{equation}

These are measurements so what we measure might be slightly different. Intuitively, the measurement error computed as ${\bf{z}}_{t} - {\bf{\hat z}}_{t}$  would be slightly different from 0.
We would call this difference ${\bf{y}} = {\bf{z}}_{t} - {\bf{\hat z}}_{t}$, the \textbf{innovation}. It represents the bias of our measurement estimation. Obviously, if everything were perfect, the innovation should be zero, as we do not have any bias! To incorporate the innovation in our model, we would add it to our state. But we would not add blindly in the state equation. We should multiply this innovation by a matrix factor that reflects somehow the correlation between our measurement bias and our state bias. Hence, we would compute a \textbf{correction} to our state equation as follows:
\begin{equation}
{\bf{\hat x}}_{t} = {\bf{F}} {\bf{x}}_{t-1}  + {\bf{B}}_{t}  {\bf{u}}_{t} + \mathbf{K}_t \bf{y}
\end{equation}

The matrix $\mathbf{K}_t$ is called the \textbf{Kalman gain}. We would see in section \ref{TheMaths} how to determine it but for the sake of intuition of Kalman filter, we skip this details. 
What really matters is to build an intuition. We can easily understand the following rules of thumb:

\begin{enumerate}
\item The more noisy our measurement is, the less precise it is. Hence, the innovation that represents the bias in our measurement may not be real innovation but rather an artifact of the measurement noise. Noise or uncertainty is directly captured by variance. Thus the larger the measurement variance noise, the lower the Kalman gain should be.
\item The more noisy our process state is, the more important the innovation should be taken into account. Hence, the larger the process state variance, the larger the Kalman gain should be.
\end{enumerate}

Summarizing these two intuitions, we would expect the Kalman gain to be:
\begin{equation}
\mathbf{K}_t \sim \frac{\text{Process Noise}}{\text{Measurement Noise}} 
\end{equation}

Fifth, noise itself should be modeled as we have no clue about real noise. Noise in our model is represented by variance, or more precisely by the covariance of our state.
Traditionally, we denote by $\mathbf{P}_t$ the covariance matrix of the state:
\begin{equation}
\mathbf{P}_t = \operatorname{Cov}({\bf{\hat x}}_t)
\end{equation}

Using the state equation, and using the fact that for any matrix $\mathbf{A}$, $\operatorname{Cov}( \mathbf{A} .\mathbf{X}) = \mathbf{A} \operatorname{Cov}(\mathbf{X}) \mathbf{A}^\top$, we can derive $\mathbf{P}_t$ from its previous state:

\begin{equation}\mathbf{P}_t = \operatorname{Cov}({\bf{\hat x}}_{t}) = \\
 \operatorname{Cov}(\mathbf{F} {\bf{x}}_{t-1}) = \\
\mathbf{F} \operatorname{Cov}({\bf{x}}_{t-1}) \mathbf{F}^\top = \\
\mathbf{F} \mathbf{P}_{t-1} \mathbf{F}^\top \label{simple_covariance_eq}
\end{equation} 

At this stage, we can make the model even more realistic. Previous equation (\ref{simple_covariance_eq}) assumes that our process model is perfect. But keep in mind that we temporarily assumed no noise. However, real world is more complex and we should now use our real state equation (\ref{state_equation}). 
In particular, if we assume that the state noise $\bf{w}_t$ is independent from the state $\bf{x}_t$ and if the state noise $\bf{w}_t$ is assumed to be distributed as a normal distribution with covariance matrix $\mathbf{Q}_t$, the prediction for our state covariance matrix becomes:

\begin{equation}
\mathbf{P}_t = \mathbf{F} \mathbf{P}_{t-1} \mathrm{F}^\top + \mathbf{Q}_{t-1}
\end{equation}

Likewise, we could compute the covariance matrix for the measurement and model the evolution of noise that we will denote by $\mathbf{S}_t$. 
The last source of noise in our system is the measurement. Following the same logic we obtain a covariance matrix for ${\bf{\hat z}}_{t}$ and denoting by ${\bf v}_t$ the normally distributed independent noise for our measurements, with covariance matrix given by $\mathbf{R}_t$, we get 

\begin{equation}
\mathbf{S}_{t} = \operatorname{Cov}({\bf{\hat z}}_{t}) = \operatorname{Cov}(\mathbf{H} {\bf{\hat x}}_{t-1} + {\bf v}_{t-1}) = \mathbf{H} \mathbf{P}_{t-1} \mathbf{H}^\top + \mathbf{R}_{t-1}
\end{equation}

Let us come back to our Kalman gain computation. Intuitively, we found that it should be larger for larger process noise and lower for large measurement noise, leading to an equation of the type 
$\mathbf{K}_t \sim \frac{\text{Process Noise}}{\text{Measurement Noise}}$. As process noise is measured $\mathbf{P}_t$ and measurement noise by $\mathbf{S}_t$, we should get something like
$ \mathbf{K}_t = \mathbf{P}_t \mathbf{S}_{t}^{-1}$.  We shall see in the next section that the real equation for the Kalman gain is closely related to our intuition and given by
\begin{equation}
\mathbf{K}_t = \mathbf{P}_t \mathbf{H}^{\top} \mathbf{S}_{t}^{-1}
\end{equation}

\section{The maths}\label{TheMaths}
In this section, we will prove rigorously all the equations provided in section \ref{TheIntuition}. 
To make notations more robust and make the difference for a time dependent vector or matrix $\mathbf{A}$ between its value at time $t$, knowing information up to time $t-1$, 
and its value at time $t$, knowing information up to time $t$, we will denote these two different value $\mathbf{A}_{t | t-1}$ and $\mathbf{A}_{t | t}$. 

\subsection{Kalman filter modeling assumption}\label{assumptions}
The Kalman filter model assumes the true state at time $t$ is obtained from the state at the previous time $t-1$ by the following state equation

\begin{equation} \label{state_equation_second_time}
\mathbf{x}_t = \mathbf{F}_t \mathbf{x}_{t-1} + \mathbf{B}_t \mathbf{u}_t + \mathbf{w}_t 
\end{equation}

The noise process $\mathbf{w}_t$ is assumed to follow a multi dimensional normal distribution with zero mean and covariance matrix given by $\mathbf{Q}_t$: $\mathbf{w}_t \sim \mathcal{N}\left(0, \mathbf{Q}_t\right)$.

At time $t$, we make a measurement (or an observation) $\mathbf{z}_t$ of the true state $\mathbf{x}_t$ according to our measurement equation:

\begin{equation}\label{measurement_equation}
\mathbf{z}_t = \mathbf{H}_t \mathbf{x}_t + \mathbf{v}_t
\end{equation}

Like for the state equation, we assume that the observation noise $\mathbf{v}_t$ follows a multi dimensional normal distribution with zero mean and covariance matrix given by $\mathbf{R}_t$: $\mathbf{v}_t \sim \mathcal{N}\left(0, \mathbf{R}_t\right)$. In addition, the initial state, and noise vectors at each step ${\mathbf{x}_0, \mathbf{w}_1, \ldots, \mathbf{w}_t, \mathbf{v}_1, \ldots, \mathbf{v}_t}$ are assumed to be all mutually independent.

\subsection{Properties}
It is immediate to derive the prediction phase as the estimated value of the state $\hat{\mathbf{x}}_{t\mid t-1}$ is simply the expectation of the state equation (\ref{state_equation}). Similarly, it is trivial to derive the estimated value for the error covariance. This provides the prediction phase that is summarized below

\begin{align}
\hat{\mathbf{x}}_{t\mid t-1} &= \mathbf{F}_t\hat{\mathbf{x}}_{t-1\mid t-1} + \mathbf{B}_t \mathbf{u}_t  		& \scriptstyle{\text{(Predicted state estimate)}}   \label{predict_eq1}\\
\mathbf{P}_{t\mid t-1} &= \mathbf{F}_t \mathbf{P}_{t-1\mid t-1} \mathbf{F}_t^\mathrm{T} + \mathbf{Q}_t 	& \scriptstyle{\text{(Predicted error covariance)}} \label{predict_eq2}
\end{align}

The correction phase that consists in incorporating the measurements to correct our prediction is slightly more tricky. It consists in the following equations:
\begin{align}
\tilde{\mathbf{y}}_t	&=\mathbf{z}_t - \mathbf{H}_t\hat{\mathbf{x}}_{t\mid t-1}  						& \scriptstyle{\text{(Measurement pre-fit residua)}} \label{correct_eq1} \\
\mathbf{S}_t 			&=\mathbf{R}_t + \mathbf{H}_t \mathbf{P}_{t\mid t-1} \mathbf{H}_t^\mathrm{T}		& \scriptstyle{\text{(Innovation covariance)}} \label{correct_eq2} \\
\mathbf{K}_t 			&=\mathbf{P}_{t\mid t-1}\mathbf{H}_t^\mathrm{T} \mathbf{S}_t^{-1}				& \scriptstyle{\text{(Optimal Kalman gain)}} \label{correct_eq3} \\
\hat{\mathbf{x}}_{t\mid t} &=\hat{\mathbf{x}}_{t\mid t-1} + \mathbf{K}_t\tilde{\mathbf{y}}_t				& \scriptstyle{\text{(Updated state estimate)}} \label{correct_eq4} \\
\mathbf{P}_{t|t} 		&=(\mathbf{I} - \mathbf{K}_t \mathbf{H}_t )\mathbf{P}_{t|t-1}						& \scriptstyle{\text{(Updated estimate covariance)}} \label{reduced_form} \\
\tilde{\mathbf{y}}_{t\mid t} &=\mathbf{z}_t - \mathbf{H}_t\hat{\mathbf{x}}_{t\mid t}						& \scriptstyle{\text{(Measurement post-fit residual)}} \label{correct_eq6} 
\end{align}

\begin{proposition}
Under the assumptions stated in subsection \ref{assumptions}, the Kalman gain that minimizes the expected squared error defined as the square of the Euclidean (or $L_2$) norm of the error vector, representing the error between the true state and our estimated state : $\mathbf{x}_t - \hat{\mathbf{x}}_{t\mid t}$ is given by 
\begin{equation}
\mathbf{K}_t 	=\mathbf{P}_{t\mid t-1}\mathbf{H}_t^\mathrm{T} \mathbf{S}_t^{-1}
\end{equation}
It is referred to as the optimal Kalman gain. For any Kalman gain (and not necessarily the optimal one), the estimate covariance updates as follows:
\begin{equation}\label{Joseph_form_0}
\mathbf{P}_{t\mid t} = (\mathbf{I} - \mathbf{K}_t \mathbf{H}_t) \mathbf{P}_{t\mid t-1} (\mathbf{I} - \mathbf{K}_t \mathbf{H}_t)^\mathrm{T} + \mathbf{K}_t \mathbf{R}_t \mathbf{K}_t^\mathrm{T} 
\end{equation}

For the optimal Kalman filter, this reduces to the usual Kalman filter Updated estimate covariance as follows:
\begin{equation}\label{reduced_form_0}
\mathbf{P}_{t|t} =(\mathbf{I} - \mathbf{K}_t \mathbf{H}_t )\mathbf{P}_{t|t-1}		
\end{equation}
\end{proposition}

\begin{proof}
The derivation of all these formulae consists in three steps. First, we derive the posteriori estimate of the covariance matrix. Then we compute the Kalman gain as the minimum square error estimator. Third, we show that in the particular case of the optimal gain the equation for the updated estimate, covariance reduces to the formula provided in equation (\ref{reduced_form}). All these steps are provided in appendix \ref{Kalman_filter_proof} to make the reading of this article smooth.
\end{proof}

\section{Kalman as Graphical model}\label{GraphicalKalman}
Historically, Hidden Markov Model (HMM) and Kalman filter were developed in distinct and unconnected research communities. Hence, their close relationship has not always been widely emphasized and appreciated. Part of the explanation lies also to the fact that the general framework for unifying these two approaches, namely graphical models came much later than HMM and Kalman filter. Without Bayesian graphical framework, the two algorithms underlying the inference calculation look rather different and unrelated. However, their difference is simply a consequence of the differences between discrete and continuous hidden variables and more specifically between multinomial and normal distribution. These details, important as they may be in practice, should not obscure us from the fundamental similarity between these two models. As we shall see in this section, the inference procedure for the state space model (SSM) shall prove us shortly that HMM and Kalman filter's model are cousin and share the same underlying graphical model structure, namely a hiddent state space variable and an observable variable. The interest of using Bayesian Probabilistic Graphical model is multiple. First, it emphasizes the general graphical model architecture underpinning both HMM and Kalman filter. Second, it provides modern computational tools used commonly in machine learning to do the inference calculation. It shows how to generalize Kalman filter in the case of non Gaussian assumptions. 
It is interesting to realize that graphical models have been the marriage between probability theory and graph theory. 

They provide a natural tool for dealing with two problems that occur throughout applied mathematics and engineering – uncertainty and complexity – and in particular they are playing an increasingly important role in the design and analysis of machine learning algorithms.\\\\

From a literature point of view, Hidden Markov models were discussed as early as \cite{Rabiner_1986}, and expanded on in \cite{Rabiner_1989}. 
The first temporal extension of probabilistic graphical models is due to \cite{Dean_1989}, who also coined the term dynamic Bayesian network. Much work has been done on defining various representations that are based on hidden Markov models or on dynamic Bayesian networks; 
these include generalizations of the basic framework, or special cases that allow more tractable inference. 
Examples include mixed memory Markov models (see \cite{Saul_1999}); variable-duration HMMs (\cite{Rabiner_1989}) and their extension segment models (\cite{Ostendorf_1996}); factorial HMMs (\cite{Ghahramani_Jordan_1994}); and hierarchical HMMs (\cite{Fine_1998} and \cite{Bui_2001}).
\cite{Smyth_1997} is a review paper that was influential in providing a clear exposition of the connections between HMMs and DBNs. 
\cite{Murphy_2001} show how hierarchical HMMs can be reduced to DBNs, 
a connection that provided a much faster inference algorithm than previously proposed for this representation. 
\cite{Murphy_2002} (and lately the book \cite{Murphy_2013}) provides an excellent tutorial on the topics 
of dynamic Bayesian networks and related representations as well as the non published book of \cite{Jordan_2016}

\subsection{State space model}
The state space model as emphasized for instance in \cite{Murphy_2013} is described as follows:
\begin{itemize}
\item there is a continuous chain of states denoted by $(\mathbf{x}_t)_{t=1, \ldots, n}$ that are non observable and influenced by past states only through the last realization. In other words, $\mathbf{x}_t$ is a Markov process, meaning 
$\mathbb{P}(\mathbf{x}_t \mid \mathbf{x}_1, \mathbf{x}_2, \ldots, \mathbf{x}_{t-1}) = \mathbb{P}(\mathbf{x}_t \mid \mathbf{x}_{t-1})$
Using graphical model, this can also be stated as given a state at one point in time, the states in the future are conditionally independent of those in the past.
\item for each state, we can observe a space variable denoted by $\mathbf{z}_t$ that depends on the non observable space $\mathbf{x}_t$
\end{itemize}

Compared to the original presentation of the Kalman filter model, this is quite different. We now assume that there is an hidden variable (our state) and we can only measure a space variable. Since at each step, the space variable only depends on the non observable, there is only one arrow or edge between two latent variables horizontal nodes. This model is represented as a graphical model in figure \ref{SSM}.\\

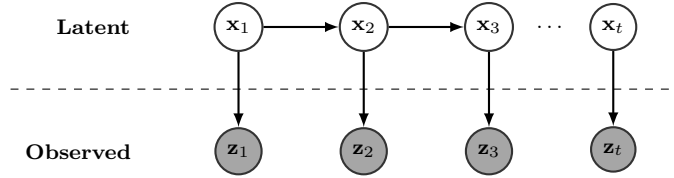
\begin{figure}[h]
\centering

\begin{tikzpicture}
 \node[box,draw=white!100] (Latent) {\textbf{Latent}};
 \node[main] (L1) [right=of Latent] {$\mathbf{x}_1$};
 \node[main] (L2) [right=of L1] {$\mathbf{x}_2$};
 \node[main] (L3) [right=of L2] {$\mathbf{x}_3$};
 \node[main] (Lt) [right=of L3] {$\mathbf{x}_t$};
 \node[main,fill=black!35] (O1) [below=of L1] {$\mathbf{z}_1$};
 \node[main,fill=black!35] (O2) [below=of L2] {$\mathbf{z}_2$};
 \node[main,fill=black!35] (O3) [below=of L3] {$\mathbf{z}_3$};
 \node[main,fill=black!35] (Ot) [below=of Lt] {$\mathbf{z}_t$};
 \node[box,draw=white!100,left=of O1] (Observed) {\textbf{Observed}};

 \path (L1) edge [connect] (L2)
        (L2) edge [connect] (L3)
        (L3) -- node[auto=false]{\ldots} (Lt);

 \path (L1) edge [connect] (O1)
	(L2) edge [connect] (O2)
	(L3) edge [connect] (O3)
	(Lt) edge [connect] (Ot);

 \draw [dashed, shorten >=-1cm, shorten <=-1cm]
      ($(Latent)!0.5!(Observed)$) coordinate (a) -- ($(Lt)!(a)!(Ot)$);
\end{tikzpicture}
\caption{State Space model as a Bayesian Probabilistic Graphical model. Each vertical slice represents a time step. Nodes in white represent unobservable or latent variables  called the states and denoted by $\mathbf{x}_t$ while nodes in gray observable ones and are called the spaces and denoted by $\mathbf{z}_t$. Each arrow indicates that there is a relationship between the arrow originating node and the arrow targeting node. Dots indicate that there is many time steps. The central dot line is to emphasize the fundamental difference between latent and observed variables} \label{SSM}
\end{figure}

Obviously, the graphical model provided in figure \ref{SSM} can host both HMM and Kalman filter model. To make our model tractable, we will need to make additional assumptions. We will emphasize the ones that are common to HMM and Kalman filter models and the ones that differ. We impose the following conditions:

\begin{itemize}
\item The relationship between the state and the space is linear (this is our measurement equation in the Kalman filter and this is common to Kalman filter and HMM models): 
\begin{equation}\label{measurement_equation2}
\mathbf{z}_t = \mathbf{H}_t \mathbf{x}_t + \mathbf{v}_t
\end{equation}
where the noise term $\mathbf{v}_t$ is assumed to follow a multi dimensional normal distribution with zero mean and covariance matrix given by $\mathbf{R}_t$.

\item We will make the simplest choice of dependency between the state at time $t-1$ and $t$ and assume that this is linear (this is our state equation and this is common to Kalman filter and HMM models)
\begin{equation} \label{state_equation2}
\mathbf{x}_t = \mathbf{F}_t \mathbf{x}_{t-1} + \mathbf{B}_t \mathbf{u}_t + \mathbf{w}_t 
\end{equation}
where the noise term $\mathbf{w}_t$ is assumed to follow a multi dimensional normal distribution with zero mean and covariance matrix given by $\mathbf{Q}_t$ and where $\mathbf{B}_t \mathbf{u}_t$ is an additional trend term (that represents our control in Kalman filter). This control term is not common in HMM model but can be added without any difficulty. This results in slightly extended formula that we will signal. These formula are slight improvement of the common one found in the literature. Although, from a theoretical point of view, this control term may seem a futility, it is very important in practice and makes a big difference in numerical applications.

\item We will assume as in the Kalman filter section that the initial state, and noise vectors at each step ${\mathbf{x}_0, \mathbf{w}_1, \ldots, \mathbf{w}_t, \mathbf{v}_1, \ldots, \mathbf{v}_t}$ are all mutually independent (this is common to HMM and Kalman filter models).

\item Last but not least, we assume that the distribution of the state variable $\mathbf{x}_t$ follows a multi dimensional normal distribution. This is \textbf{Kalman filter specific}. For HMM, the state is assumed to follow a multinomial distribution.
\end{itemize}

The above assumptions restrict our initial state space model to a narrower class called the Linear-Gaussian SSM (LG-SSM). This model has been extensively studied and more can be found in \cite{Durbin_2012} for instance.

Before embarking into the inference problem for the SSM, it is interesting to examine the unconditional distribution of the states $\mathbf{x}_{t}$. Using equation (\ref{state_equation2})

The unconditional mean of $\mathbf{x}_{t}$ is computed recursively 
\begin{equation} \label{unconditional_mean}
\prod_{k=2}^t \mathbf{F}_{k} \mathbf{x}_{1} + \sum_{k=2}^t  \prod_{l=k+1}^t \mathbf{F}_{l}  \mathbf{B}_k 
\end{equation}
with the implicit assumption that empty product equals 1: 

$$\prod_{l=t+1}^t \mathbf{F}_{l} = 1.$$

In the specific case of a null control term ($\mathbf{B}_t  = 0$), the latter equation simplifies into
\begin{equation} \label{unconditional_mean_2}
\prod_{k=2}^t \mathbf{F}_{k} \mathbf{x}_{1} 
\end{equation}

The unconditional co-variance is computed as follows $ \mathbf{P}_t = \mathbb{E}[ \mathbf{x}_t \mathbf{x}_t^\mathrm{T}]$. 
Using our assumptions on independence as well as the state equation (\ref{state_equation2}), we can compute it easily as:

\begin{eqnarray} 
\mathbf{P}_{t}  & =& \mathbf{F}_{t} \mathbf{P}_{t-1} \mathbf{F}_{t}^\mathrm{T}+ \mathbf{Q}_t
\end{eqnarray}

This last equation remains unchanged in case of a non zero control term as the control term is deterministic. This last equation provides a dynamic equation for the unconditional variance and is referred to as the \textit{Lyapunov equation}. It is also easy to checked that the unconditional covariance between neighboring states $ \mathbf{x}_{t}$ and $ \mathbf{x}_{t+1}$ is given by $\mathbf{F}_{t+1} \mathbf{P}_{t} \mathbf{F}_{t+1}^\mathrm{T}$.

\subsection{Inference}
The inference problem consists in calculating the posterior probability of the states given an output sequence. This calculation can be done both forward and backward. By forward, we mean that the inference information (called the evidence) at time $t$ consists of the partial sequence of outputs up to time $t$. The backward problem is similar except that the evidence consists of the partial sequence of outputs after time $t$. Using standard graphical model terminology, we distinguish between filtering and smoothing problem.

In filtering, the problem is to calculate an estimate of the state $\mathbf{x}_{t}$ based on a partial output sequence $\mathbf{z}_{0}, \ldots, \mathbf{z}_{t}$. That is, we want to calculate $\mathbb{P}(\mathbf{x}_{t} \mid \mathbf{z}_{0}, \ldots, \mathbf{z}_{t})$. This is often referred to as the alpha recursion in HMM models (see for instance \cite{Rabiner_1986} and \cite{Rabiner_1989}).

Using standard Kalman filter notations, we shall denote by
\begin{eqnarray}
\hat{\mathbf{x}}_{t\mid t} 	& \triangleq & \mathbb{E}[\mathbf{x}_{t} \mid \mathbf{z}_{0}, \ldots, \mathbf{z}_{t}] \\
\mathbf{P}_{t\mid t} 		& \triangleq & \mathbb{E}[(\mathbf{x}_{t}-\hat{\mathbf{x}}_{t\mid t} ) (\mathbf{x}_{t}-\hat{\mathbf{x}}_{t\mid t} )^\mathrm{T} \mid \mathbf{z}_{0}, \ldots, \mathbf{z}_{t}]
\end{eqnarray}

Under this settings, it is fairly easy to derive the following property that provides the conditional posterior distribution in the forward recursion. And to recover traditional results of Kalman filter, we shall decompose our time propagation into two steps:
\begin{itemize}
\item time update: 		\qquad	\qquad \qquad	$ \mathbb{P}[\mathbf{x}_{t} \mid \mathbf{z}_{0}, \ldots, \mathbf{z}_{t}] \,\,\,\,\,\, \rightarrow \mathbb{P}[\mathbf{x}_{t+1} \mid \mathbf{z}_{0}, \ldots, \mathbf{z}_{t}] $
\item measurement update:		\quad \,\,\,\,\,	$ \mathbb{P}[\mathbf{x}_{t+1} \mid \mathbf{z}_{0}, \ldots, \mathbf{z}_{t}] \rightarrow \mathbb{P}[\mathbf{x}_{t+1} \mid \mathbf{z}_{0}, \ldots, \mathbf{z}_{t+1}] $
\end{itemize}

This can be represented nicely in terms of graphical models by the figure \ref{fig_two_steps} below.

\begin{figure}[h]
\centering

\begin{tikzpicture}
 \node[main, label= {$\mathbf{x}_{t}$}] (L1) []  {};
 \node[main, label={$\mathbf{x}_{t+1}$}] (L2) [right=of L1] {};
 \node[main,fill=black!35,  label=below:{$\mathbf{z}_{t}$}] (O1) [below=of L1] {}; 
 \node[main,fill=white!100,  label=below:{$\mathbf{z}_{t+1}$}] (O2) [below=of L2] {};

 \path (L1) edge [connect] (L2)
        (L2) edge [connect] (O2);

 \path (L1) edge [connect] (O1);
\end{tikzpicture}
\qquad \qquad \qquad \qquad 
\begin{tikzpicture}
 \node[main, label= {$\mathbf{x}_{t}$}] (L1) []  {};
 \node[main, label={$\mathbf{x}_{t+1}$}] (L2) [right=of L1] {};
 \node[main,fill=black!35,  label=below:{$\mathbf{z}_{t}$}] (O1) [below=of L1] {}; 
 \node[main,fill=black!35,  label=below:{$\mathbf{z}_{t+1}$}] (O2) [below=of L2] {};

 \path (L1) edge [connect] (L2)
        (L2) edge [connect] (O2);

 \path (L1) edge [connect] (O1);
\end{tikzpicture}\\
\textbf{(a)} \qquad \qquad \qquad \qquad \qquad \qquad \qquad \textbf{(b)}
\caption{(a) A portion of the State Space Model before a measurement and (b) after a measurement update. White nodes are non observable variables while gray nodes are observed nodes.} \label{fig_two_steps}
\end{figure}
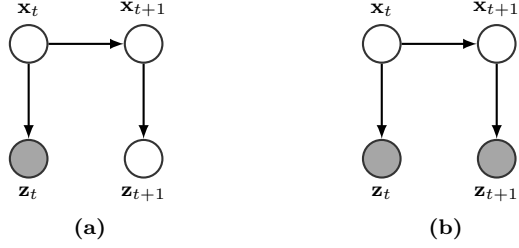

\begin{proposition}
Conditioned on past outputs $\mathbf{z}_{0}, \ldots, \mathbf{z}_{t}$, the variables $\mathbf{x}_{t+1}$ and $\mathbf{z}_{t+1}$ have a joint Gaussian distribution with mean and covariance matrix given by:
\begin{equation}
\left[ 
\begin{array}{c}
\hat{\mathbf{x}}_{t+1\mid t} \\
\mathbf{H}_{t+1} \hat{\mathbf{x}}_{t+1\mid t} 
\end{array}
\right]
\quad
\text{and}
\quad
\left[ 
\begin{array}{c c}
\mathbf{P}_{t+1\mid t}  & \mathbf{P}_{t+1\mid t} \mathbf{H}_{t+1}^\mathrm{T} \\
\mathbf{H}_{t+1} \mathbf{P}_{t+1\mid t} & \mathbf{H}_{t+1} \mathbf{P}_{t+1\mid t} \mathbf{H}_{t+1}^\mathrm{T} +\mathbf{R}_{t+1}
\end{array}
\right]
\end{equation}
\end{proposition}

\begin{proof}
This is trivial as the considered state space model is the Linear Gauss State Space Model (LGSSM).
\end{proof}

As simple as it may seem, the previous proposition makes our graphical model a full powerhouse as it provides the building block to start the inference. Indeed, knowing the conditional posterior distribution at step (a) of figure \ref{fig_two_steps} gives us the first bullet to conclude for the step (b). Moreover, using results from simpler graphical models like the factor analysis graphical model, we can immediately conclude that the second step is given as follows

\begin{proposition}\label{Graphical_model_prop1}
Conditioned on past outputs $\mathbf{z}_{0}, \ldots, \mathbf{z}_{t+1}$, we have the following relationship
\begin{eqnarray}
\mkern-20mu \hat{\mathbf{x}}_{t+1\mid t+1}&=& \hat{\mathbf{x}}_{t+1\mid t} + \mathbf{P}_{t+1\mid t} \mathbf{H}_{t+1}^\mathrm{T} \left( \mathbf{H}_{t+1} \mathbf{P}_{t+1\mid t} \mathbf{H}_{t+1}^\mathrm{T} + \mathbf{R}_{t+1} \right)^{-1} \nonumber \\
& & \qquad \qquad \qquad \qquad \qquad \qquad  \qquad \qquad  ( \mathbf{z}_{t+1}- \mathbf{H}_{t+1} \hat{\mathbf{x}}_{t+1\mid t} ) \label{prop_Kalman_graph_eq1} \\
\mkern-20mu  \mathbf{P}_{t+1\mid t+1}&=& \mathbf{P}_{t+1\mid t} - \mathbf{P}_{t+1\mid t} \mathbf{H}_{t+1}^\mathrm{T} \left( \mathbf{H}_{t+1} \mathbf{P}_{t+1\mid t} \mathbf{H}_{t+1}^\mathrm{T} + \mathbf{R}_{t+1} \right)^{-1} \mathbf{H}_{t+1} \mathbf{P}_{t+1\mid t}  \quad \quad  \quad \label{prop_Kalman_graph_eq2}
\end{eqnarray}
\end{proposition}

\begin{proof}
This is a direct consequence of the spatial model and can be found for instance in \cite{Murphy_2013} or \cite{Jordan_2016}. We provide a self contained proof in appendix \ref{factor_analysis_proof}
\end{proof}

Summarizing all these results leads to the seminal recursions of the Kalman filter as provided in the section \ref{TheMaths} and given by the following proposition. 
\begin{proposition}
Kalman filter consists in the following recursive equations:
\begin{align}
\hat{\mathbf{x}}_{t\mid t-1} &= \mathbf{F}_t\hat{\mathbf{x}}_{t-1\mid t-1} + \mathbf{B}_t \mathbf{u}_t  		&\label{predict_eq1_2}\\
\mathbf{P}_{t\mid t-1} &= \mathbf{F}_t \mathbf{P}_{t-1\mid t-1} \mathbf{F}_t^\mathrm{T} + \mathbf{Q}_t 	&\label{predict_eq2_2} \\
\hat{\mathbf{x}}_{t+1\mid t+1} &= \hat{\mathbf{x}}_{t+1\mid t} + \mathbf{P}_{t+1\mid t} \mathbf{H}_{t+1}^\mathrm{T} \left( \mathbf{H}_{t+1} \mathbf{P}_{t+1\mid t} \mathbf{H}_{t+1}^\mathrm{T} + \mathbf{R}_{t+1} \right)^{-1} \nonumber \\
&  \qquad \qquad \qquad \qquad \qquad \qquad  \qquad \qquad  ( \mathbf{z}_{t+1}- \mathbf{H}_{t+1} \hat{\mathbf{x}}_{t+1\mid t} ) \label{prop_Kalman_graph_eq1_2} \\
\mathbf{P}_{t+1\mid t+1}&= \mathbf{P}_{t+1\mid t} - \mathbf{P}_{t+1\mid t} \mathbf{H}_{t+1}^\mathrm{T} \left( \mathbf{H}_{t+1} \mathbf{P}_{t+1\mid t} \mathbf{H}_{t+1}^\mathrm{T} + \mathbf{R}_{t+1} \right)^{-1} \mathbf{H}_{t+1} \mathbf{P}_{t+1\mid t}    &\label{prop_Kalman_graph_eq2_2}
\end{align}
\end{proposition}

\begin{proof}
This is a trivial consequence of proposition \ref{Graphical_model_prop1}. Equation (\ref{predict_eq1_2}) (resp. (\ref{predict_eq2_2})) is the same as the one provided in (\ref{predict_eq1}) (resp. (\ref{predict_eq2})) but using graphical models induction.
\end{proof}

\begin{rem}
If we introduce the following intermediate variables already provided in section \ref{TheMaths}:
\begin{align}
\tilde{\mathbf{y}}_{t+1}	&=\mathbf{z}_{t+1} - \mathbf{H}_{t+1} \hat{\mathbf{x}}_{t+1 \mid t}  							&\label{correct_eq1_2} \\
\mathbf{S}_{t+1} 			&=\mathbf{R}_{t+1} + \mathbf{H}_{t+1} \mathbf{P}_{t+1\mid t} \mathbf{H}_{t+1}^\mathrm{T}	&\label{correct_eq2_2} \\
\mathbf{K}_{t+1}			&=\mathbf{P}_{t+1\mid t}\mathbf{H}_{t+1}^\mathrm{T} \mathbf{S}_{t+1}^{-1}					&\label{correct_eq3_2} &&
\end{align}
The equations (\ref{prop_Kalman_graph_eq1_2}) and (\ref{prop_Kalman_graph_eq2_2}) transform into equations
\begin{align}
\hspace{-1cm} \hat{\mathbf{x}}_{t+1\mid t+1} &=\hat{\mathbf{x}}_{t+1\mid t} + \mathbf{K}_{t+1}\tilde{\mathbf{y}}_{t+1}    		\hspace{1cm}		& \\
\hspace{-1cm} \mathbf{P}_{t+1\mid t+1}&= \left( \mathbf{I} - \mathbf{K}_{t+1} \mathbf{H}_{t+1} \right) \mathbf{P}_{t+1\mid t}    	\hspace{1cm}		&&
\end{align}
which are equations (\ref{correct_eq4}) and (\ref{reduced_form}). This proves that the derivation using graphical models and control theory are mathematically equivalent!
\end{rem}
\vspace{0.3cm}

\begin{rem}
There is nothing new to fancy at this stage except that we have shown with graphical models the Kalman filter recursion. And we can check this is way faster, easier and more intuitive. It is worth noticing that we have also multiple ways to write the gain matrix. Using the Sherman–Morrison–Woodbury formula, we can also derive various forms for the gain matrix as follows:
\small
\begin{align}
\mkern-18mu \mathbf{K}_{t+1} &=\mathbf{P}_{t+1\mid t} \mathbf{H}_{t+1}^\mathrm{T} \left( \mathbf{H}_{t+1} \mathbf{P}_{t+1\mid t} \mathbf{H}_{t+1}^\mathrm{T} + \mathbf{R}_{t+1} \right)^{-1} \\
&= \left( \mathbf{P}_{t+1\mid t}^{-1} + \mathbf{H}_{t+1}^\mathrm{T} \mathbf{R}_{t+1} \mathbf{H}_{t+1} \right)^{-1} \mathbf{H}_{t+1}^\mathrm{T} \mathbf{R}_{t+1}^{-1} \\
&= \left( \mathbf{P}_{t+1\mid t} - \mathbf{P}_{t+1\mid t} \mathbf{H}_{t+1}^\mathrm{T} ( \mathbf{H}_{t+1} \mathbf{P}_{t+1\mid t} \mathbf{H}_{t+1}^\mathrm{T} + \mathbf{R}_{t+1})^{-1} \mathbf{H}_{t+1} \mathbf{P}_{t+1\mid t}\right) \nonumber \\
&  \qquad \qquad \qquad \qquad \qquad \qquad  \qquad \qquad  \qquad\qquad\qquad \qquad\qquad \mathbf{H}_{t+1}^\mathrm{T} \mathbf{R}_{t+1}^{-1}   \\
&=  \mathbf{P}_{t+1\mid t+1} \mathbf{H}_{t+1}^\mathrm{T} \mathbf{R}_{t+1}^{-1} \label{last_equation_gain}
\end{align}
\normalsize
These forms may be useful whenever the reduced form (which is the last equation) is numerically unstable. Equation (\ref{last_equation_gain}) is useful as it relates $\mathbf{K}_{t+1} $ to $ \mathbf{P}_{t+1\mid t+1}$.  It is interesting to notice the two form of the Kalman gain

 $$
\mathbf{K}_{t+1}	=\mathbf{P}_{t+1\mid t}\mathbf{H}_{t+1}^\mathrm{T} \mathbf{S}_{t+1}^{-1}	=  \mathbf{P}_{t+1\mid t+1} \mathbf{H}_{t+1}^\mathrm{T} \mathbf{R}_{t+1}^{-1} 
$$
\end{rem}
\vspace{0.3cm}

\begin{rem}
The Kalman filter appeals some remarks. We can first note that Kalman filtering equations can be interpreted differently. Combining equations (\ref{correct_eq4}) and (\ref{predict_eq1}), we retrieve an error correcting algorithm as follows:
$$
\hat{\mathbf{x}}_{t\mid t} = \mathbf{F}_{t} \hat{\mathbf{x}}_{t-1\mid t-1} + \mathbf{B}_t  \mathbf{u}_t + \mathbf{K}_{t} \left( z_{t}- \mathbf{H}_{t} ( \mathbf{F}_{t} \hat{\mathbf{x}}_{t-1\mid t-1} + \mathbf{B}_t \mathbf{u}_t  ) \right)
$$

or equivalently, regrouping the term $ \hat{\mathbf{x}}_{t-1\mid t-1}$
$$
\hat{\mathbf{x}}_{t\mid t} = \left( \mathbf{F}_{t} -  \mathbf{K}_{t} \mathbf{H}_{t} \mathbf{F}_{t} \right)   \hat{\mathbf{x}}_{t-1\mid t-1} +  \left(  \mathbf{B}_t  \mathbf{u}_t  + \mathbf{K}_{t} \left( z_{t} - \mathbf{H}_{t} \mathbf{B}_t \mathbf{u}_t  \right)\right)
$$

This shows us two things:
\begin{itemize}
\item Kalman filter can be seen as the discrete version of an Ornstein Uhlenbeck process
\item Kalman filter can be seen as Auto Regressive process in discrete times
\end{itemize}

Since, recursive equation do appear similarly in Recursive Least Square (RLS) estimates, we also see here a connection with RLS. It is striking that these connections are not often made, mostly because Kalman filter was originally a control problem and not a statistician one.
\end{rem}

As nice as the Kalman filter equation may look like, they have one major problem. It is the numerical stability of the filter.  If the process noise covariance $ \mathbf{Q}_t$ is small, round-off error would lead to obtain numerically a negative number for small positive eigenvalues of this matrix. As the scheme will propagate round off errors, the state covariance matrix $\mathbf{P}_t$ will progressively become only positive semi-definite (and hence indefinite) while it is theoretically a true positive definite matrix fully invertible.

In order to keep the positive definite property, we can slightly modify the recursive equation to find recursive equations that preserve the positive definite feature of the two covariance matrices. Since any positive definite matrix $\mathbf{S}^d$ can be represented and also reconstructed by its upper triangular square root matrix $\mathbf{R}$ with $\mathbf{S}^d = \mathbf{R}^t \mathbf{R}^\mathrm{T}$, with the strong property that representing in this form will guarantee that the resulting matrix will never get a negative eigen value, it is worth using a square root scheme with square root representation. Alternatively, we can also represent our covariance matrix using the so called unit diagonal (U-D) decomposition form, with 
$\mathbf{S}^d = \mathbf{U} \mathbf{D} \mathbf{U}^\mathrm{T}$ where $\mathbf{U}$' is a unit triangular matrix (with unit diagonal), and $\mathbf{D}$' is a diagonal matrix. This form avoids in particular many of the square root operations required by the matrix square root representation. Moreover, when comparing the two approaches, the U-D form has the elegant property to need same amount of storage, and somewhat less computation. This explains while the U-D factorization is often preferred \cite{Thornton_1979}. A slight variation is the 
\textit{LDL} decomposition of the innovation covariance matrix. The LDL decomposition relies on decomposing the covariance matrix $\mathbf{S}^d$ with two matrices:
a lower unit triangular (unitriangular) matrix $\mathbf{L}$, and $\mathbf{D}$ a diagonal matrix. The algorithm starts with the LU decomposition as implemented in the Linear Algebra PACKage (LAPACK) and further it factors into the LDL form. Any singular covariance matrix is pivoted so that the first diagonal partition is always non-singular and well-conditioned (for further details see \cite{Bar-Shalom_2002}).

\subsection{Connection to information filter}
It is worth showing the close connection with particle and information filter. This is a direct consequence of the fact that a multivariate Gaussian belongs to the exponential family and as such admits canonical parameters. Hence, we can rewrite the filter in terms of the later instead of the initial usage of moment parameters. This is an interesting rewriting of the Kalman filter as it makes it numerically more stable. The canonical parameters of a multi variate Gaussian distribution, denoted by $\mathbf{\Lambda}$ and $\eta$, are obtained from the moment parameters $\Sigma$ and $\mu$ as follows: $\mathbf{\Lambda} = \Sigma^{-1}$ and $\eta=\Sigma^{-1} \mu$. We are interested in deriving the canonical parameters of $\mathbf{x}_{t}$ first,  at the prediction phase, conditioned on $\mathbf{z}_{1}, \ldots, \mathbf{z}_{t-1}$. 

In the Kalman filter settings, the covariance $\Sigma$ is denoted by $\mathbf{P}$ while the moment is given by $\mathbf{x}$ with the relationship $\eta = \mathbf{P}^{-1} \mathbf{x}$. Hence, we will write the precision matrix as $\mathbf{\Lambda}$ with the relationship with the covariance matrix given by $\mathbf{\Lambda}=\mathbf{P}^{-1}$. We shall write our new variables to be consistent with previous development as $\mathbf{\Lambda}_{t \mid t-1}$ and $\eta_{t \mid t-1}$. At the correction or measurement phase, we are interested in the same parameters but now conditioned on  $\mathbf{z}_{1}, \ldots, \mathbf{z}_{t}$. We shall write them  $\mathbf{\Lambda}_{t \mid t}$ and $\eta_{t \mid t}$. We can easily derive the following recursive scheme that is given by the proposition below:

\begin{proposition}\label{information_filter}
The filter equations are given by the following recursive scheme:
\begin{eqnarray}
\hat{\eta}_{t \mid t-1} &=& \mathbf{Q}_{t}^{-1} \mathbf{F}_{t} (\mathbf{\Lambda}_{t-1 \mid t-1} +\mathbf{F}_{t}^\mathrm{T}  \mathbf{Q}_{t}^{-1} \mathbf{F}_{t})^{-1}\hat{\eta}_{t-1 \mid t-1} \nonumber \\
& &\!\!\!\!\! + (\mathbf{Q}_{t}^{-1} \! \!- \! \mathbf{Q}_{t}^{-1} \mathbf{F}_{t} ( \mathbf{\Lambda}_{t-1 \mid t-1} \!+ \! \mathbf{F}_{t}^\mathrm{T} \mathbf{Q}_{t}^{-1} \mathbf{F}_{t})^{-1} \mathbf{F}_{t}^\mathrm{T} \mathbf{Q}_{t}^{-1} ) \mathbf{B}_t  \mathbf{u}_t \\
\hat{\eta}_{t \mid t}&=&\hat{\eta}_{t \mid t-1}+\mathbf{H}_{t}^\mathrm{T}\mathbf{R}_{t}^{-1} \mathbf{z}_{t}\\
\mathbf{\Lambda}_{t \mid t-1} &=&  \mathbf{Q}_t^{-1}-\mathbf{Q}_t^{-1} \mathbf{F}_{t} (\mathbf{\Lambda}_{t-1 \mid t-1} + \mathbf{F}_{t}^\mathrm{T} \mathbf{Q}_t^{-1} \mathbf{F}_{t})^{-1}  \mathbf{F}_{t}^\mathrm{T} \mathbf{Q}_t^{-1} \\
\mathbf{\Lambda}_{t \mid t} &=&  \mathbf{\Lambda}_{t \mid t-1} +\mathbf{H}_{t}^\mathrm{T} \mathbf{R}_{t}^{-1}\mathbf{H}_{t} 
\end{eqnarray}
\end{proposition}

\begin{proof}
The derivation is easy and given in section \ref{filter_equation_proof}
\end{proof}

\begin{rem}
The recursive equation for the information filter are very general. They include the control term $ \mathbf{B}_t  \mathbf{u}_t $ that is often neglected in literature introduced as early as 1979 in \cite{Anderson_1979}. It is worth noticing that we can simplify computation by pre-computing a term $\mathbf{M}_t$ as follows:
\begin{eqnarray}
\mathbf{M}_t & = & \mathbf{Q}_{t}^{-1} \mathbf{F}_{t} (\mathbf{\Lambda}_{t-1 \mid t-1} +\mathbf{F}_{t}^\mathrm{T}  \mathbf{Q}_{t}^{-1} \mathbf{F}_{t})^{-1} \\
\mathbf{\Lambda}_{t \mid t-1} &=&  \mathbf{Q}_t^{-1} - \mathbf{M}_t \mathbf{F}_{t}^\mathrm{T} \mathbf{Q}_t^{-1} \\
\hat{\eta}_{t \mid t-1} &=& \mathbf{M}_t  \hat{\eta}_{t-1 \mid t-1} + \mathbf{\Lambda}_{t \mid t-1} \mathbf{B}_t  \mathbf{u}_t \\
\hat{\eta}_{t \mid t}&=&\hat{\eta}_{t \mid t-1}+\mathbf{H}_{t}^\mathrm{T}\mathbf{R}_{t}^{-1} \mathbf{z}_{t}\\
\mathbf{\Lambda}_{t \mid t} &=&  \mathbf{\Lambda}_{t \mid t-1} +\mathbf{H}_{t}^\mathrm{T} \mathbf{R}_{t}^{-1}\mathbf{H}_{t} 
\end{eqnarray}

These equations are more efficient than the ones provided in proposition \ref{information_filter}. As for the initialization of this recursion, we define the initial value as follows $\hat{\eta}_{1 \mid 0}= \hat{\eta}_{1}$ and $\mathbf{\Lambda}_{1 \mid 0} =\mathbf{\Lambda}_{1}$.
It is interesting to note that the Kalman filter and the information filter are mathematically equivalent. They both share the same assumptions. However, they do not use the same parameters. Kalman filter (KF) uses moment parameters while particle or information filter (IF) relies on canonical parameters, which makes the later numerically more stable in case of poor conditioning of the covariance matrix. This is easy to understand as a small eigen value of the covariance translates into a large eigen value of the precision matrix as the precision matrix is the inverse of the covariance matrix. Reciprocally, a poor conditioning of the information filter should convert to a more stable scheme for the Kalman filter as the he condition number of a matrix is the reciprocal of the condition number of its inverse. Likewise, for initial condition as they are inverse, a small initial state in KF should translate to a large state in IF.
\end{rem}

\subsection{Smoothing}
Another task we can do on dynamic Bayesian network is smoothing. It consists in obtaining estimates of the state at time $t$ based on information from $t$ on-wards (we are using future information in a sense). Like for HMM, the computation of this state estimate requires combining forward and backward recursion, either by starting by a backward-filtered estimates and then a forward-filtered estimates (an  'alpha-beta algorithm"), or by doing an algorithm  that iterates directly on the filtered-and-smoothed estimates (an  "alpha-gamma  algorithm").  Both  kinds of algorithm are available in the literature on state-space models (see for instance \cite{Koller_2009} and \cite{Cappe_2010} for more details), but the latter approach appears to dominate (as opposed to the HMM literature, where the former approach dominates) . The "alpha-gamma" approach is referred to as the "Rauch-Tung-Striebel (RTS) smoothing algorithm" (although it was developed using very different tool, namely control theory as early as 1965: see \cite{Rauch_1965}) while the other approach is just the "alpha-beta" recursion.

\subsubsection{Rauch-Tung-Striebel (RTS) smoother}
The Rauch-Tung-Striebel (RTS) smoother developed in \cite{Rauch_1965} relies precisely on the idea of of doing first a backward estimate and then a forward filter. Smoothing should not be confused with smoothing in times series analysis that is more or less a convolution. Smoothing for a Bayesian network means infering the distribution of a node conditioned on future information. It is the reciprocal of filtering that has also two meanings. Filtering for time series means doing a convolution to filter some noise. But filtering for Bayesian network means infering the distribution of a node conditioned on past information. \\

We can work out the RTS smoother and find the recursive equations provided by the following proposition

\begin{proposition}
The RTS smoothing algorithm works as follows:
\begin{flalign}
 \hat{\mathbf{x}}_{t\mid T}  &=  \hat{\mathbf{x}}_{t\mid t}  +  \mathbf{L}_{t}(\mathbf{x}_{t+1 \mid T} - \hat{\mathbf{x}}_{t+1\mid t} )  & \\
 \hat{\mathbf{P}}_{t\mid T}  &=  \mathbf{P}_{t \mid t} + \mathbf{L}_{t} (\mathbf{P}_{t+1 \mid T}-\mathbf{P}_{t+1 \mid t}) \mathbf{L}_{t}^\mathrm{T}  &\\
\text{where }  \mathbf{L}_{t}  &= \mathbf{P}_{t \mid t} \mathbf{F}_{t+1} ^\mathrm{T} \mathbf{P}_{t+1 \mid t}^{-1}  &
\end{flalign}

\noindent with an initial condition given by
\begin{eqnarray}
\hat{\mathbf{x}}_{T \mid T} & = & \hat{\mathbf{x}}_{T} \qquad \qquad\qquad\\
\hat{\mathbf{P}}_{T \mid T} & = & \hat{\mathbf{P}}_{T} \qquad\qquad\qquad\qquad\qquad\qquad
\end{eqnarray}
\end{proposition}

\begin{proof}
The proof consists in writing rigorously the various equations and is given in section \ref{RTS_recursive_equation_proof}
\end{proof}

\subsubsection{Alternative to RTS filter}
It is worth noting that we can develop an alternative approach to the RTS filter that relies on the alpha beta approach without any observation. 
This approach is quite standard for HMM models (see for instance \cite{Rabiner_1986} or \cite{Russell_2009}). Following \cite{Kitagawa_1987}, this approach has been called in the literature the \textit{two-filter algorithm}. It consists in combining the \textit{forward} conditional probability $\mathbb{P}(\mathbf{x}_{t} \mid \mathbf{z}_{1}, \ldots, \mathbf{z}_{t})$ with the \textit{backward} conditional probability $\mathbb{P}(\mathbf{x}_{t} \mid \mathbf{z}_{t+1}, \ldots, \mathbf{z}_{T})$. The intuition is illustrated by the figure \ref{fig_no_observations}.

\Large
\begin{figure}[h]
\centering
\begin{tikzpicture}
 \node[main, draw=white!100] (N1)  {\Huge{\ldots}};
 \node[main, fill=white!100, label={$\mathbf{x}_{t}$}] (N2)  [right=of N1] {};
 \node[main, fill=white!100, label={$\mathbf{x}_{t+1}$}] (N3)  [right=of N2] {};
 \node[main, draw=white!100] (N4)  [right=of N3] {\Huge{\ldots}};
 \node[main,fill=white!100,  label=below:{$\mathbf{z}_{t}$}] (O2) [below=of N2] {};
 \node[main,fill=white!100,  label=below:{$\mathbf{z}_{t+1}$}] (O3) [below=of N3] {};

 \path (N1) edge [connect] (N2)  
         (N2) edge [connect] (N3)
         (N3) edge [connect] (N4);
 \path (N2) edge [connect] (O2);
 \path (N3) edge [connect] (O3);
 \draw[edge] (N2) -- (N3) node[midway, above] {$\mathbf{F}_{t+1}$};
\end{tikzpicture}

\normalsize
\caption{SSM  with no observations} \label{fig_no_observations}
\end{figure}
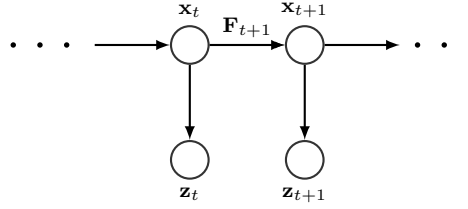
\normalsize

The graphical model provided by figure \ref{fig_no_observations} can be easily characterized. The joint probability distribution of $(\mathbf{x}_{t}, \mathbf{x}_{t+1})$ is a multi variate normal whose Lyapunov equation (equation of the covariance matrix) is given by

\begin{eqnarray}\label{initial_equation_two_filter}
\mathbf{P}_{t+1} = \mathbf{F}_{t+1} \mathbf{P}_t \mathbf{F}_{t+1}^\mathrm{T} + \mathbf{Q}_{t+1}
\end{eqnarray}

Hence the covariance matrix of $(\mathbf{x}_{t}, \mathbf{x}_{t+1})$ is given by

\begin{equation}
\left[ 
\begin{array}{c  c}
\mathbf{P}_t & \mathbf{P}_t \mathbf{F}_{t+1}^\mathrm{T}  \\
\mathbf{F}_{t+1} \mathbf{P}_t  &  \mathbf{F}_{t+1} \mathbf{P}_t  \mathbf{F}_{t+1}^\mathrm{T} +\mathbf{Q}_{t+1}
\end{array}
\right]
\end{equation}

The underlying idea in this approach is to invert the arrows in the graphical model leading to a new graphical model given by
\Large
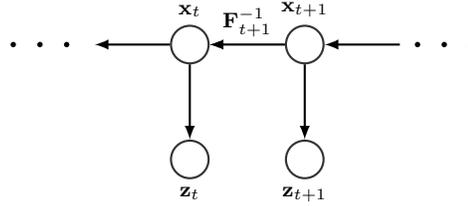
\begin{figure}[h]
\centering
\begin{tikzpicture}
 \node[main, draw=white!100] (N1)  {\Huge{\ldots}};
 \node[main, fill=white!100, label={$\mathbf{x}_{t}$}] (N2)  [right=of N1] {};
 \node[main, fill=white!100, label={$\mathbf{x}_{t+1}$}] (N3)  [right=of N2] {};
 \node[main, draw=white!100] (N4)  [right=of N3] {\Huge{\ldots}};
 \node[main,fill=white!100,  label=below:{$\mathbf{z}_{t}$}] (O2) [below=of N2] {};
 \node[main,fill=white!100,  label=below:{$\mathbf{z}_{t+1}$}] (O3) [below=of N3] {};

 \path (N2) edge [connect] (N1)  
         (N3) edge [connect] (N2)
         (N4) edge [connect] (N3);
 \path (N2) edge [connect] (O2);
 \path (N3) edge [connect] (O3);
 \draw[edge] (N3) -- (N2) node[midway, above] {$\mathbf{F}_{t+1}^{-1}$};
\end{tikzpicture}

\normalsize
\caption{Inverted arrows for the SSM  with no observations} \label{inverted_graph}
\end{figure}
\normalsize

Mathematically, this imply to change the relationship and now solve for $\mathbf{P}_{t}$ in terms of $\mathbf{P}_{t+1}$ using equation (\ref{initial_equation_two_filter}). We get:

\begin{eqnarray}
\mathbf{P}_{t} &= &\mathbf{F}_{t+1}^{-1} ( \mathbf{P}_{t+1}  - \mathbf{Q}_{t+1} ) \mathbf{F}_{t+1}^\mathrm{-T} 
\end{eqnarray}

where we have assumed that $\mathbf{F}_{t+1}$ is invertible. As a matter of fact, if it is not the case, we can rewrite everything with an approaching matrix to $\mathbf{F}_{t+1}$ (in the sense of the Frobenius norm) that is invertible, whose distance is less than $\varepsilon$, derive all the relationship below and then take the limit as $\epsilon$ tends to zero. We will therefore assume in the following that $\mathbf{F}_{t+1}$ is invertible as this is not a strict condition for our derivation. Hence we can rewrite the covariance matrix now as follows:

\begin{equation}
\left[ 
\begin{array}{c  c}
\mathbf{F}_{t+1}^{-1} ( \mathbf{P}_{t+1}  - \mathbf{Q}_{t+1} ) \mathbf{F}_{t+1}^\mathrm{-T}  &  \mathbf{F}_{t+1}^{-1} ( \mathbf{P}_{t+1}  - \mathbf{Q}_{t+1} )  \\
 ( \mathbf{P}_{t+1}  - \mathbf{Q}_{t+1} )  \mathbf{F}_{t+1}^{-T} &  \mathbf{P}_{t+1} 
\end{array}
\right]
\end{equation}

If we define the new matrix $\widetilde{\mathbf{F}}_{t+1}$ as:
\begin{equation}
\widetilde{\mathbf{F}}_{t+1} = \mathbf{F}_{t+1}^{-1} ( \mathbf{I}-   \mathbf{Q}_{t+1} \mathbf{P}_{t+1}^\mathrm{-1} ) 
\end{equation}

we can simplify the covariance matrix as follows:
\begin{equation}
\left[ 
\begin{array}{c  c}
\widetilde{\mathbf{F}}_{t+1} \mathbf{P}_{t+1}  \mathbf{F}_{t+1}^{-T} &  \widetilde{\mathbf{F}}_{t+1} \mathbf{P}_{t+1}  \\
\mathbf{P}_{t+1}  \widetilde{\mathbf{F}}_{t+1}^\mathrm{T} &  \mathbf{P}_{t+1} 
\end{array}
\right]
\end{equation}

As this looks like a forward covariance matrix. Recall that the forward dynamics is defined by:
\begin{eqnarray}\label{forward_dynamics}
\mathbf{x}_{t+1}  & =&  \mathbf{F}_{t+1} \mathbf{x}_{t} + \mathbf{B}_{t+1} \mathbf{u}_{t+1}+ w_{t+1}
\end{eqnarray}

The following proposition gives the inverse dynamics:

\begin{proposition}
The inverse dynamics is given by:
\begin{eqnarray}\label{inverse_dynamics}
\mathbf{x}_{t}  & =& \widetilde{\mathbf{F}}_{t+1} \mathbf{x}_{t+1}  + \widetilde{\mathbf{B}}_{t+1} \mathbf{u}_{t+1}  +\widetilde{w}_{t+1}
\end{eqnarray} 
\begin{eqnarray}
\text{with} \hspace{0.5cm} \widetilde{w}_{t+1} &=&  -  \mathbf{F}_{t+1}^{-1} (w_{t+1} - \mathbf{Q}_{t+1}  \mathbf{P}_{t+1}^{-1} \mathbf{x}_{t+1})  \hspace{3.5cm} \\
\widetilde{\mathbf{B}}_{t+1} &=& - \mathbf{F}_{t+1}^{-1} \mathbf{B}_{t+1}
\end{eqnarray}

and with $\widetilde{w}_{t+1}$ independent of the \textit{past} information $\mathbf{x}_{t+1}, \ldots, \mathbf{x}_{T}$ and with the covariance matrix of $\widetilde{w}_{t+1} $ given by
\begin{eqnarray} \label{inverse_dynamics_cov}
\widetilde{\mathbf{Q}}_{t+1} &\triangleq &\mathbb{E}[ \widetilde{w}_{t+1} \widetilde{w}_{t+1}^\mathrm{T} ] = \mathbf{F}_{t+1}^{-1} \mathbf{Q}_{t+1} (\mathbf{I} - \mathbf{P}_{t+1}^{-1} \mathbf{Q}_{t+1}) \mathbf{F}_{t+1}^{-\mathrm{T}},
\end{eqnarray}
and with the forward Lyapunov equation given by:
\begin{eqnarray}
\mathbf{P}_t &=& \widetilde{\mathbf{F}}_{t} \mathbf{P}_{t+1} \widetilde{\mathbf{F}}_{t}^\mathrm{T} +  \widetilde{\mathbf{Q}}_{t+1} 
\end{eqnarray}
\end{proposition}

\begin{proof}
The proof is straightforward and given in \ref{proof_inverse_dynamics}.
\end{proof}

The reasoning followed here gives us a new way to generate an information filter but now backward. We should note that the conversion between the state space and the observable variable is unchanged and given by
$$
{\bf{z}}_t = {\bf{H}} {\bf{x}}_t +{\bf{v}}_t
$$

We can now apply the information filter developed previously to get the new filtering equations using canonical instead of moment parameterization. We get the following new information filter :
\begin{proposition}
The modified Bryson–Frazier as presented in \cite{Bierman_2006} is given by:
\begin{eqnarray}
\widetilde{\mathbf{M}}_t & = &  \mathbf{F}_{t}^{T}  \mathbf{Q}_{t}^{-1} (\mathbf{\Lambda}_{t+1 \mid t+1} +\ \mathbf{Q}_{t}^{-1} - \mathbf{P}_{t}^{-1})^{-1} \hspace{1.5cm} \\
\mathbf{\Lambda}_{t \mid t+1} &=&  \mathbf{F}_{t}^{T} (\mathbf{Q}_t - \mathbf{Q}_t \mathbf{P}_{t}^{-1} \mathbf{Q}_{t})^{-1}   \mathbf{F}_{t}  - \widetilde{\mathbf{M}}_t \mathbf{Q}_{t}^{-1} \mathbf{F}_{t} \\
\hat{\eta}_{t \mid t+1} &=& \widetilde{\mathbf{M}}_t  \hat{\eta}_{t+1 \mid t+1} - \mathbf{\Lambda}_{t \mid t+1}  \mathbf{F}_{t} \mathbf{B}_t  \mathbf{u}_t \\
\hat{\eta}_{t \mid t}&=&\hat{\eta}_{t \mid t+1}+\mathbf{H}_{t}^\mathrm{T}\mathbf{R}_{t}^{-1} \mathbf{z}_{t}\\
\mathbf{\Lambda}_{t \mid t} &=&  \mathbf{\Lambda}_{t \mid t+1} +\mathbf{H}_{t}^\mathrm{T} \mathbf{R}_{t}^{-1}\mathbf{H}_{t} 
\end{eqnarray}
\end{proposition}

\begin{proof}
The proof consists in combining all previous results and is given in \ref{proof_Bryson_Frazier}.
\end{proof}

\begin{rem}
If we want to convert to to the moment representation, we can use the inverse transform given by $\hat{\mathbf{x}}_{t\mid t+1}  = \mathbf{S}_{t\mid t+1}^{-1} \hat{\mathbf{\eta}}_{t\mid t+1}$ and 
$\mathbf{P}_{t\mid t+1}  = \mathbf{S}_{t\mid t+1}^{-1}$
\end{rem}

\subsection{Inferring final posterior distribution}
This problem consists in inferring the final posterior distribution $\mathbb{P}(\mathbf{x}_t \mid \mathbf{z}_1,  \ldots, \mathbf{z}_T)$. 
We can easily derive 
\begin{equation}
\hat{\mathbf{x}}_{t\mid T} = \mathbf{P}_{t\mid T} (  \mathbf{P}_{t\mid t} ^{-1} \hat{\mathbf{x}}_{t\mid t} + \mathbf{P}_{t\mid t+1} ^{-1} \hat{\mathbf{x}}_{t\mid t+1})
\end{equation}

and
\begin{equation}
\mathbf{P}_{t\mid T} = \left(\mathbf{P}_{t\mid t}^{-1} + \mathbf{P}_{t\mid t+1}^{-1} - \Sigma_{t}^{-1} \right)^{-1}
\end{equation}

\subsection{Parameter estimation}

\subsubsection{Intuition}
Since the connection between HMM and  Kalman filter, it has been established that the best way to estimate the initial parameters of a Kalman filter and any of its extension is to rely on EM estimation. 

The expectation–maximization (EM) algorithm is an iterative method to find the maximum likelihood or maximum a posteriori (MAP) estimates of the parameters of a statistical model.  
This approach relies on the fact that the model depends on non observable (also called latent) variables.  The EM algorithm alternates between performing an expectation (E) step, which provides the expectation of the conditional log-likelihood evaluated using the current estimate for the parameters, and a maximization (M) step, which computes parameters maximizing the expected log-likelihood found on the E step. These parameter-estimates are then used to determine the distribution of the latent variables in the next E step.

The EM algorithm was explained and given its name in a classic 1977 paper by \cite{Dempster_1977}. The intuition is to find a way to solve the maximum likelihood solution for a model with latent variables. Since variables are hidden, solving the maximum likelihood solution typically requires taking the derivatives of the likelihood function with respect to all the unknown values and then solving the resulting equations. Because of latent variables, this function is unknown and the problem can not be solved. Instead, we compute the expected maximum likelihood with respect to the latent variables. We can then solve this explicit maximum likelihood function with respect to the observable variables.

\subsubsection{EM Algorithm}
We assume that our statistical model has a set $\mathbf{X}$ of observed data and a set of unobserved latent $\mathbf{Z}$ that depend on unknown parameters $\boldsymbol\theta$ that we want to determine thanks to maximum likelihood. We obviously know the likelihood function $L(\boldsymbol\theta; \mathbf{X}, \mathbf{Z}) = p(\mathbf{X}, \mathbf{Z}|\boldsymbol\theta)$ given by

$$
L(\boldsymbol\theta; \mathbf{X}) = p(\mathbf{X}|\boldsymbol\theta) = \int  p(\mathbf{X},\mathbf{Z}|\boldsymbol\theta) d\mathbf{Z}
$$

The later is often intractable as the number of values for $\mathbf{Z}$ is vary large making the computation of the expectation (the integral) extremely difficult. In order to avoid the exact computation, we can work as follows:

First compute the expectation explicitly (called the \textit{E step}) by taking the expected value of the log likelihood function of $\boldsymbol\theta$  with respect to latent variables $\mathbf{Z}$, denoted by $Q(\boldsymbol\theta| \boldsymbol\theta^{(t)}$:

$$Q(\boldsymbol\theta|\boldsymbol\theta^{(t)}) = \operatorname{E}_{\mathbf{Z}|\mathbf{X},\boldsymbol\theta^{(t)}}\left[ \log L (\boldsymbol\theta;\mathbf{X},\mathbf{Z})  \right] \,$$

We then maximize the resulted function with respect to the parameters $\boldsymbol\theta$:
$$
\boldsymbol\theta^{(t+1)} = \underset{\boldsymbol\theta}{\operatorname{arg\,max}} \ Q(\boldsymbol\theta|\boldsymbol\theta^{(t)}) \, 
$$

The EM method works both for discrete or continuous random latent variables. It takes advantage of algoirthm like the Viterbi algorithm for hidden Markov models or recursive equation fo Kalman filter. The various steps are the following:

\begin{algorithm}
\caption{EM algorithm:}
\begin{algorithmic} 
\State \textbf{init:}  Set $\boldsymbol\theta$  to some random values.
\While {Not Converged}
	\State Compute the probability of each possible value of $\mathbf{Z}$, given $\boldsymbol\theta$ (E-step)
	\State  Use the computed values of $\mathbf{Z}$ to find the argmax values for $\boldsymbol\theta$ (M-Step).
\EndWhile
\end{algorithmic}
\end{algorithm}

The convergence of this algorithm is given by the following proposition.

\begin{proposition}
The EM algorithm converges monotonically to a local minimum for the marginal log likelihood.
\end{proposition}

\begin{proof} There are at least two ways of proving this results that are provided in \ref{EMProof1} and \ref{EMProof2}.
\end{proof}

\subsection{Application to Kalman filter}
The EM algorithm was initially developed for mixture models in particular Gaussian mixtures but also other natural laws from the exponential family such as Poisson, binomial, multinomial and exponential distributions as early as in \cite{Hartley_1958}. It was only once the link between latent variable and Kalman filter models was made that it became obvious that this could also be applied to Kalman and extended Kalman filter (see \cite{Cappe_2010} or \cite{Einicke_2010}). The EM method works as folllows:

\begin{algorithm}[H]
\caption{EM algorithm for Kalman filter}
\begin{algorithmic} 
\State \textbf{init:}  Set the Kalman filter parameters to some random values.
\While{Not Converged}
	\State Use Kalman filter recursive equations to obtained updated state estimates (E-step)
	\State Use the filtered or smoothed state estimates within maximum-likelihood calculations to obtain updated parameter estimates.
	For instance for the standard Kalman filter, the distributions are assumed to be normal. The corresponding maximum likehood estimate for the updated measurement noise variance 	estimate leads to a variance given by
\begin{equation*}
\hat{\sigma}^{2}_v = \frac{1}{N} \sum_{k=1}^N {(z_k-\hat{x}_{k})}^{2}
\end{equation*}
where $\hat{x}_k$ are the output estimates calculated by the Kalman filter for the  measurements ${z_k}$. 
Likewise, the the updated process noise variance estimate is calculated as
\begin{equation*}
\hat{\sigma}^{2}_w =   \frac{1}{N} \sum_{k=1}^N {(\hat{x}_{k+1}-\hat{F}\hat{{x}}_{k})}^{2}
\end{equation*}

where $\hat{x}_k$ and $\hat{x}_{k+1}$ are scalar state estimates calculated by a filter or a smoother. 
The updated model coefficient estimate is obtained via
\begin{equation*}\hat{F} = \frac{\sum_{k=1}^N (\hat{x}_{k+1}-\hat{F} \hat{x}_k)}{\sum_{k=1}^N \hat{x}_k^{2}} 
\end{equation*}.
\EndWhile
\end{algorithmic}
\end{algorithm}

\subsubsection{CMA-ES estimation}
Another radically difference approach is to minimize some cost function depending on the Kalman filter parameters. As opposed to the maximum likelihood approach that tries to find the best suitable distribution that fits the data, this approach can somehow factor in some noise and directly target a cost function that is our final result. Because our model is an approximation of the reality, this noise introduction may leads to a better overall cost function but a worse distribution in terms of fit to the data. 

Let us first introduce the CMA-ES algorithm. Its name stands for covariance matrix adaptation evolution strategy. As it points out, it is an evolution strategy optimization method, meaning that it is a derivative free method that can accommodate non convex optimization problem. The terminology covariance matrix alludes to the fact that the exploration of new points is based on a multinomial distribution whose covariance matrix is progressively determined at each iteration. Hence the covariance matrix adapts in a sense to the sampling space, contracts in dimension that are useless and expands in dimension where natural gradient is steep. This algorithm has led to a large number of papers and articles and we refer to \cite{Hansen_2018}, \cite{Ollivier_2017}, \cite{Auger_2016}, \cite{Auger_2015}, \cite{Hansen_2014}, \cite{Auger_2012}, \cite{Hansen_2011}, \cite{Auger_2009}, \cite{Igel_2007}, \cite{Auger_2004} to cite a few of the numerous articles around CMA-ES. We also refer the reader to the excellent Wikipedia page \cite{wiki:CMAES}.

CMA-ES relies on two main principles in the exploration of admissible solution for our optimization problem.
First, it relies on a multi variate normal distribution as this is the maximum entropy distribution given the first two moments.
The mean of the multi variate distribution is updated at each step in order to maximize the likelihood of finding a successful candidate. The second moment, the covariance matrix of the distribution is also updated at each step to increase the likelihood of successful search steps. These updates can be interpreted as a natural gradient descent. Intuitively, the CMA ES algorithm conducts an iterated principal components analysis of successful search steps while retaining all principal axes. 

Second, we retain two paths of the successive distribution mean, called search or evolution paths. The underlying idea is keep significant information about the correlation between consecutive steps. If consecutive steps are taken in a similar direction, the evolution paths become long. The evolution paths are exploited in two ways. We use the first path is to compute the covariance matrix to increase variance in favorable directions and hence increase convergence speed. The second path is used to control step size and to make consecutive movements of the distribution mean orthogonal in expectation. The goal of this step-size control is to prevent premature convergence yet obtaining fast convergence to a local optimum.

In order to make it practical, we assume that we have a general cost function that depends on our Bayesian graphical model denoted by $\Phi( \theta)$ where $\theta$ are the parameters of our Kalman filter. Our cost function is for instance the Sharpe ratio corresponding to a generic trend detection strategy whose signal is generated by our Bayesian graphical model that is underneath a Kalman filter. This approach is more developed in a companion paper \cite{Benhamou_2018_CMAES} but we will give here the general idea. Instead of computing the parameter of our Bayesian graphical model using the EM approach, we would like to find the parameters $\theta_{\max}$ that maximize our cost function $\Phi( \theta)$. Because our cost function is to enter a long trade with a predetermined target level and a given stop loss whenever our Bayesian graphical model anticipates a price risen and similarly to enter a short trade whenever our prediction based on Bayesian graphical model is a downside movement, our trading strategy is not convex neither smooth. It is a full binary function and generates spike whenever there is a trade. Moreover, our final criterium is to use the Sharpe ratio of the resulting trading strategy to compare the efficiency of our parameters. This is way too complicated for traditional optimization method, and we need to rely on Black box optimization techniques like CMA-ES. Before presenting results, we will also discuss in the following section computational issues and the choice of the State space model dynamics.

We will describe below the most commonly used CMA -ES algorithm referred to as the $\mu / \mu_{w}, \lambda$ version. It is the version in which at each iteration step, we take a weighted combination of the $\mu$ best out of $\lambda$  new candidate solutions to update the distribution parameters. The algorithm has three main loops: first, it samples new solutions, second, it reorders the sampled solutions based on their fitness and third it updates the state variables. A pseudo code of the algorithm is provided just below:

\begin{algorithm}[H]
\caption{CMA ES algorithm}
	\begin{algorithmic} 
	\State \textbf{Set} $\lambda$								\Comment{number of samples /  iteration}
	\State Initialize $m, \sigma, C=\mathbf{I}_n, p_\sigma=0$, $p_c=0$  \Comment{initialize state variables}
	\\
	\While{(not terminated)}  								
     		\For{$i = 1 \text{ to } \lambda$} 					\Comment{samples $\lambda$ new solutions and evaluate them}
        		\State $x_i \sim \mathcal{N}(m, \sigma^2 C )$ 	\Comment{samples multi variate normal}
	        	\State $f_i =f(x_i)$								\Comment{evaluates}
		\EndFor
		\\
	      	\State $x_{1...\lambda} = $  $x_{s(1)...s(\lambda)}$ with $s(i)$ = argsort($f_{1...\lambda}$, $i$)  \Comment{reorders samples}
     		\State $m' = m$  									\Comment{stores current value of $m$}
		\\
	     	\State $m = $ update mean$(x_1, ... ,$ $x_\lambda)$  \Comment{udpates mean to better solutions}
     		\State $p_\sigma = $ update ps$(p_\sigma,$ $\sigma^{-1} C^{-1/2} (m - m'))$  \Comment{updates isotropic evolution path}
	     	\State $p_c = $ update pc$(p_c,$ $\sigma^{-1}(m - m'),$ $||p_\sigma||)$  \Comment{updates anisotropic evolution path}
	     	\State $C = $ update C$(C,$ $p_c,$ ${(x_1 - m')}/{\sigma},... ,$ ${(x_\lambda - m')}/{\sigma})$  \Comment{updates covariance matrix}
	     	\State $\sigma = $ update sigma$(\sigma,$ $||p_\sigma||)$  \Comment{updates step-size using isotropic path length}
		\\
		\State not terminated $=$ iteration $\leq$ iteration max and $||m-m'|| \geq \varepsilon$ \Comment{stop condition}
	\EndWhile
	\\ \\
	\Return $m$ or $x_1$			\Comment{returns solution}
	\end{algorithmic}
\end{algorithm}

In this algorithm, for an $n$ dimensional optimization program, the five state variables at iteration step $k$ are:
\begin{enumerate}
\item $m_k\in\mathbb{R}^n$, the distribution mean and current best solution,

\item $\sigma_k>0$, the variance step-size,

\item $ C_k$, a symmetric positive-definite $n\times n$ covariance matrix initialized to $\mathbf{I}_n$,

\item $ p_\sigma\in\mathbb{R}^n$, the isotropic evolution path, initially set to null,

\item $p_c\in\mathbb{R}^n$, the anisotropic evolution path, initially set to null.
\end{enumerate}

It is worth noting that the order of the five update is important. The iteration starts with sampling $\lambda>1$ candidate solutions $x_i \in\mathbb{R}^n $ 
from a multivariate normal distribution $\textstyle \mathcal{N}(m_k,\sigma_k^2 C_k)$, that is $x_i  \sim\ m_k + \sigma_k \mathcal{N}(0,C_k) $ for  $i=1,...,\lambda$.

We then evaluate candidate solutions $ x_i$ for the objective function $f:\mathbb{R}^n\to\mathbb{R}$ of our optimization.  

We then sort candidate solution according to their objective function value: 
$ \{x_{i:\lambda}\;|\;i=1\dots\lambda\} \text{ with } f(x_{1:\lambda})\le \ldots \le f(x_{\mu:\lambda}) \le \ldots \le f(x_{\lambda:\lambda})$. It is worth noticing that we do not even need to know the value. Only the ranking is important for this algorithm. 

The new mean value is updated as a weighted mean of our new candidate solutions
\begin{align}
  m_{k+1} &= \sum_{i=1}^{\mu} w_i\, x_{i:\lambda}  = m_k + \sum_{i=1}^{\mu} w_i\, (x_{i:\lambda} - m_k) 
\end{align}

where the positive (recombination) weights $ w_1 \ge w_2 \ge \dots \ge w_\mu > 0$ sum to one. Typically, $\mu \le \lambda/2$ and the weights are chosen such that $\textstyle \mu_w := 1 / \sum_{i=1}^\mu w_i^2 \approx \lambda/4$. 

The step-size $\sigma_k$ is updated using cumulative step-size adaptation (CSA), sometimes also denoted as path length control. The evolution path $p_\sigma$ is updated first as it is used in the update of the step-size $\sigma_k$:

\begin{align}
  p_\sigma & = \underbrace{(1-c_\sigma)}_{\!\!\!\!\!\text{discount factor}\!\!\!\!\!}\, p_\sigma 
    + \overbrace{\sqrt{1 - (1-c_\sigma)^2}}^{
     \!\!\!\!\!\!\!\!\!\!\!\!\!\!\!\!\!\!\!\!\!\!\!\!\!\!\!\!\!\!\!\text{complements for discounted variance}
     \!\!\!\!\!\!\!\!\!\!\!\!\!\!\!\!\!\!\!\!\!\!\!\!\!\!\!\!\!\!\!} \underbrace{\sqrt{\mu_w} 
     \,C_k^{\;-1/2} \, \frac{\overbrace{m_{k+1} - m_k}^{\!\!\!\text{displacement of}\; m\!\!\!}}{\sigma_k}}_{\!\!\!\!\!\!\!\!\!\!\!\!\!\!\!\!\!\!
                      \text{distributed as}\; \mathcal{N}(0,I)\;\text{under neutral selection}
                      \!\!\!\!\!\!\!\!\!\!\!\!\!\!\!\!\!\!} \\
  \sigma_{k+1} & = \sigma_k \times \exp\bigg(\frac{c_\sigma}{d_\sigma}
                          \underbrace{\left(\frac{\|p_\sigma\|}{E\|\mathcal{N}(0,I)\|} - 1\right)}_{\!\!\!\!\!\!\!\!\!\!\!\!\!\!\!\!\!\!\!\!\!\!\!\!\!\!\!\!\!\!\!\!\!\!\!\!
    \text{unbiased about 0 under neutral selection}
    \!\!\!\!\!\!\!\!\!\!\!\!\!\!\!\!\!\!\!\!\!\!\!\!\!\!\!\!\!\!\!\!\!\!\!\!
}\bigg)
\end{align}

where 
\begin{itemize}
\item $c_\sigma^{-1}\approx n/3$ is the backward time horizon for the evolution path $p_\sigma$ and larger than one ($c_\sigma \ll 1$ is reminiscent of an exponential decay constant as $(1-c_\sigma)^k\approx\exp(-c_\sigma k)$ where $c_\sigma^{-1}$ is the associated lifetime and $c_\sigma^{-1}\ln(2)\approx0.7c_\sigma^{-1}$ the half-life),

\item $\mu_w=\left(\sum_{i=1}^\mu w_i^2\right)^{-1}$ is the variance effective selection mass and $1 \le \mu_w \le \mu$ by definition of $w_i$,

\item $C_k^{\;-1/2} = \sqrt{C_k}^{\;-1} = \sqrt{C_k^{\;-1}}$ is the unique symmetric square root of the inverse of $C_k$, and

\item $d_\sigma$ is the damping parameter usually close to one. For $d_\sigma=\infty$ or $c_\sigma=0$ the step-size remains unchanged.
\end{itemize}

The step-size $\sigma_k$ is increased if and only if $\|p_\sigma\|$ is larger than $ \sqrt{n}\,(1-1/(4\,n)+1/(21\,n^2)) $
 and decreased if it is smaller (see \cite{Hansen_2006} for more details).

Finally, the covariance matrix is updated, where again the respective evolution path is updated first.

\begin{equation}
  p_c \gets \underbrace{(1-c_c)}_{\!\!\!\!\!\text{discount factor}\!\!\!\!\!}\, 
            p_c + 
     \underbrace{\mathbf{1}_{[0,\alpha\sqrt{n}]}(\|p_\sigma\|)}_{\text{indicator function}} 
     \overbrace{\sqrt{1 - (1-c_c)^2}}^{
     \!\!\!\!\!\!\!\!\!\!\!\!\!\!\!\!\!\!\!\!\!\!\!\!\!\!\!\!\!\!\!\text{complements for discounted variance}
     \!\!\!\!\!\!\!\!\!\!\!\!\!\!\!\!\!\!\!\!\!\!\!\!\!\!\!\!\!\!\!}
     \underbrace{\sqrt{\mu_w} 
     \, \frac{m_{k+1} - m_k}{\sigma_k}}_{\!\!\!\!\!\!\!\!\!\!\!\!\!\!\!\!\!\!\!\!\!\!\!\!\!\!\!\!\!\!\!\!\!\!\!\!
                      \text{distributed as}\; \mathcal{N}(0,C_k)\;\text{under neutral selection}
                      \!\!\!\!\!\!\!\!\!\!\!\!\!\!\!\!\!\!\!\!\!\!\!\!\!\!\!\!\!\!\!\!\!\!\!\!}
\end{equation}

\begin{equation}
  C_{k+1} = \underbrace{(1 - c_1 - c_\mu + c_s)}_{\!\!\!\!\!\text{discount factor}\!\!\!\!\!}
               \, C_k + c_1 \underbrace{p_c p_c^T}_{
   \!\!\!\!\!\!\!\!\!\!\!\!\!\!\!\!
   \text{rank one matrix}
   \!\!\!\!\!\!\!\!\!\!\!\!\!\!\!\!} 
         + \,c_\mu \underbrace{\sum_{i=1}^\mu w_i \frac{x_{i:\lambda} - m_k}{\sigma_k} 
             \left( \frac{x_{i:\lambda} - m_k}{\sigma_k} \right)^T}_{
                     \text{rank} \;\min(\mu,n)\; \text{matrix}}
\end{equation}

where $T$ denotes the transpose and

$c_c^{-1}\approx n/4$ is the backward time horizon for the evolution path $p_c$ and larger than one,

$\alpha\approx 1.5$ and the indicator function $\mathbf{1}_{[0,\alpha\sqrt{n}]}(\|p_\sigma\|)$ evaluates to one if and only if $\|p_\sigma\|\in[0,\alpha\sqrt{n}]$ or, in other words, $\|p_\sigma\|\le\alpha\sqrt{n}$, which is usually the case,

$c_s = (1 - \mathbf{1}_{[0,\alpha\sqrt{n}]}(\|p_\sigma\|)^2) \,c_1 c_c (2-c_c) $ makes partly up for the small variance loss in case the indicator is zero,

$c_1 \approx 2 / n^2$ is the learning rate for the rank-one update of the covariance matrix and

$c_\mu \approx \mu_w / n^2 $ is the learning rate for the rank-$\mu$ update of the covariance matrix and must not exceed $1 - c_1$.

The covariance matrix update tends to increase the likelihood function for $p_c$ and for $(x_{i:\lambda} - m_k)/\sigma_k$ to be sampled from $\mathcal{N}(0,C_{k+1})$. This completes the iteration step. The stop condition is quite standard and makes sure we stop if the best solution does not move or if we reached the maximum iterations number.

\begin{rem}
The number of candidate samples per iteration, $\lambda$, is an important and influential parameter. It highly depends on the objective function and can be tricky to be determined. Smaller values, like 10, tends to do more local search. Larger values, like 10 times the dimension $n$ makes the search more global. A potential solution to the determination of the number of candidate samples per iteration is to repeatedly restart the algorithm with increasing $\lambda$ by a factor of two at each restart (see \cite{Auger_2005})

Besides of setting $\lambda$ (or possibly $\mu$ instead, if for example $\lambda$ is predetermined by the number of available processors), the above introduced parameters are not specific to the given objective function. This makes this algorithm quite powerful as the end user has relatively few parameters to set to do an efficient optimization search.
\end{rem}

\subsection{Computational issues}
Although the computational cost of Kalman filter is very small and its complexity dammed simple, there are room for improvements.
The two dominant costs in the Kalman filter are the matrix inversion to compute the Kalman gain matrix, $\mathbf{K}_{t}$, which takes $O(| \mathbf{x}_{t} |^3 )$ time where $| \mathbf{x}_{t} |$ stands for the dimension of $\mathbf{x}_{t}$; and the matrix-matrix multiplication to compute the estimated covariance matrix which takes $O(|\mathbf{z}_{t} |^2)$ time.  \\

In some applications (e.g., robotic mapping), we have  $|\mathbf{z}_{t} \gg | \mathbf{x}_{t} |$,
 so the latter cost dominates. However, in such cases, we can sometimes use sparse approximations (see for instance \cite{Thrun_2005}).
In the opposite cases where $| \mathbf{x}_{t} | \gg |\mathbf{z}_{t}$, 
we can precompute the Kalman filter gain  $\mathbf{K}_{t}$, since, it does not depend on
the actual observations $(\mathbf{z}_{1}, \ldots, \mathbf{z}_{t})$. This is a an unusual property that is specific to linear Gaussian systems. \\

The iterative equations for updating the estimate covariance in equation (\ref{reduced_form}) or in (\ref{Joseph_form}) are Riccati equations as they are first order ordinary differential recursive equations quadratic in Kalman gain. Standard theory state that for time invariant systems, they converge to a fixed point. Hence, it may be wised to use this steady state solution instead of using a time-specific gain matrix in our filtering. In practice however, this does not provide better results. More sophisticated implementations of the Kalman filter should be used, for reasons of numerical stability. One approach is the information filter also called particle filter, which recursively updates the canonical parameters of the underlying Gaussian distribution, namely the precision matrix defined as the inverse of the covariance matrix $\Lambda_{t\mid t} = \mathbf{P}_{t\mid t}$ and the first order canonical parameter given by $\eta_{t \mid t} = \Lambda_{t\mid t} \mathbf{x}_{t\mid t}$. These parameters instead of the traditional moment parameters makes the inference stronger.

Another approach is the square root filter, which works with the Cholesky decomposition or the $UDU$ factorization of the estimated error covariance: $\mathbf{P}_{t\mid t}$. This is much more numerically stable than directly updating the estimated error covariance $\mathbf{P}_{t\mid t}$. Further details can be found at the reference site \href{http://www.cs.unc.edu/~welch/kalman/}{http://www.cs.unc.edu/~welch/kalman/} and the refernce book for Kalman filter \cite{Simon_2006}.

\section{In practice}\label{InPractice}
Financial markets are complex dynamics models. Notoriously, one does not know the exact dynamics as opposed to a noise filtering problem for a GPS. We are doomed to make assumptions about the dynamics of financial systems. This can seriously degrade the filter performance as we are left with some un-modeled dynamics. As opposed to noise that the filter can accommodate (it has been designed for it), non modeled dynamics depends on the input. They can bring the estimation algorithm to instability and divergence. It is a tricky and more experimental issue to distinguish between measurement noise and non modeled dynamics.

Kalman filter has been applied by many authors (see \cite{Lautier_2003}, \cite{lautier_2004}, \cite{Bierman_2006}, \cite{Roncalli_2011}, \cite{Dao_2011} \cite{Chan_2013} and \cite{Benhamou_2017}).

The following linear dynamics have been suggested:
\begin{align}
   \mathbf{x}_{t+1} & = \Phi    \mathbf{x}_t + c_t + w_t \\
   \mathbf{z}_t        & = \mathbf{H} \mathbf{x}_t + d_t + v_t
\end{align}   

with the following choice for the various parameters:

\begin{table}[H]
  \centering
    \begin{tabular}{| c | c| c | c | l | c | c |}
    \toprule
    \multicolumn{1}{|l|}{Model} 
	& \multicolumn{1}{l|}{ $\Phi$ } 
	& \multicolumn{1}{l|}{ $\mathbf{H}$ } 
	& \multicolumn{1}{l|}{ $\mathbf{Q}$ }
	& $\mathbf{R}$     
	& \multicolumn{1}{l|}{ $\mathbf{P}_{t=0}$ } 
	& \multicolumn{1}{l|}{ $c_t$ } \\
    \midrule
    	$1$  
	& $  \left[ {\begin{array}{cc} 1 & dt \!  \\  0 &  1 \! \end{array} } \right] $ 
        & $  \left[ {\begin{array}{c} \! 1 \!  \\  \!  0 \!  \end{array} } \right]  $ 
	& $   \left[ {\begin{array}{l}  \!  p_1^2  \; p_1 p_2 \!  \\ \!  p_1p_2 \;   p_3^2  \!  \end{array} } \right] $ 
        & $ \left[ {\begin{array}{c} \! p_4 \!  \end{array} } \right]  $ 
	& $ \left[ {\begin{array}{l}  \!  p_5  \; 0 \!  \\ \!  0 \;   p_5 \!  \end{array} } \right] $ 
	& \multicolumn{1}{c|}{$0$} \\
    \midrule
    	$2$
	& $  \left[ {\begin{array}{cc} 1 & dt \!  \\  0 &  1 \! \end{array} } \right] $ 
        & $  \left[ {\begin{array}{c} \! 1 \!  \\  \!  0 \!  \end{array} } \right]  $ 
	& $  \left[ {\begin{array}{l}  \!  p_1^2  \; p_1p_2 \!  \\ \!  p_1p_2 \;   p_3^2  \!  \end{array} } \right] $ 
        & $ \left[ {\begin{array}{c} \! p_4 \!  \end{array} } \right]  $ 
	& $ \left[ {\begin{array}{l}  \!  p_5  \; 0 \!  \\ \!  0 \;   p_6 \!  \end{array} } \right] $ 
        & \multicolumn{1}{c|}{$0$} \\
    \midrule
	$3$
	& $  \left[ {\begin{array}{cc} \!\!  p_1 & p_2 \! \!  \\ \!\!   0 &  p_3 \!\!   \end{array} } \right] $ 
        & $  \left[ {\begin{array}{c} \!\! p_4 \!\!  \\  \!\!  p_5 \! \! \end{array} } \right]  $ 
	& $   \left[ {\begin{array}{l}  \!  p_6^2  \; p_6 p_7 \!  \\ \!  p_7 p_6 \;   p_8^2  \!  \end{array} } \right] $ 
        & $ \left[ {\begin{array}{c} \! p_9 \!  \end{array} } \right]  $ 
	& $ \left[ {\begin{array}{l}  \!\!  p_{10}  \; 0 \!\!  \\ \!\! 0 \;   p_{11}  \! \! \end{array} } \right] $ 
        & \multicolumn{1}{c|}{$0$} \\
    \midrule
    	$4$     
	& $  \left[ {\begin{array}{cc} \!\!  p_1 & p_2 \! \!  \\ \!\!   0 &  p_3 \!\!   \end{array} } \right] $ 
        & $  \left[ {\begin{array}{c} \!\! p_4 \!\!  \\  \!\!  p_5 \! \! \end{array} } \right]  $ 
	& $   \left[ {\begin{array}{l}  \!  p_6^2  \; p_6p_7 \!  \\ \!  p_7 p_6 \;   p_8^2  \!  \end{array} } \right] $ 
        & $ \left[ {\begin{array}{c} \! p_9 \!  \end{array} } \right]  $ 
	& $ \left[ {\begin{array}{l}  \!\!  p_{10}  \; 0 \!\!  \\ \!\! 0 \;   p_{11}  \! \! \end{array} } \right] $ 
	& $  \left[ {\begin{array}{c} \!\!  p_{12} ( p_{13} - K_t  )\! \!  \\ \!\!  p_{14}  ( p_{15} - K_t  )\!\!  \end{array} } \right] $  \\
    \bottomrule
    \end{tabular}%
  \label{tab:addlabel}%
   \bigskip
\caption{Various Kalman filter model specifications}\label{KFModels}
\end{table}

\section{Numerical experiments}
In order to test our results, we look at the following trend following algorithm based on our Kalman filter algorithm where we enter a long trade if the prediction of our kalman filter is above the close of the previous day and a short trade if the  the prediction of our kalman filter is below the close of the previous day. For each comparison, we add an offset $\mu$ to avoid triggering false alarm signals. We set for each trade a pre-determined profit and stop loss target in ticks. These parameters are optimized in order to provide the best sharpe ratio over the train period together with the Kalman filter parameters. 
The pseudo code of our algorithm is listed below

\begin{algorithm}[H]
\caption{Kalman filter Trend following algorithm}
	\begin{algorithmic} 
	\State \textbf{Initialize common trade details}
	\State SetProfitTarget( target)							\Comment{fixed profit target in ticks}
	\State SetStopLoss( stop\_loss )							\Comment{fixed stop loss in ticks}
	\\
	\While{ Not In Position}											\Comment{look for new trade}
		\If{ KF( $p_1, \ldots, p_n$).Predict[0] $\ge$ Close[0] + $\mu$} 	\Comment{up trend signal}
			\State EnterLong()										\Comment{market order for the open}
		\ElsIf{ KF( $p_1, \ldots, p_n$).Predict[0] $\le$ Close[0] + $\mu$} 	\Comment{down trend signal}
			\State EnterShort()										\Comment{market order for the open}
		\EndIf
	\EndWhile
	\end{algorithmic}
\end{algorithm}

Our resulting algorithm depends on the following parameters $p_1, \ldots, p_n$ the Kalman filter algorithm, the profit target, the stop loss and the signal offset $\mu$. We could estimate the Kalman filter parameters with the EM procedure, then optimize the profit target, the stop loss and the signal offset $\mu$. However, if by any chance the dynamics of the Kalman filter is incorrectly specified, the noise generated by this wrong specification will only be factored in the three parameters:  the profit target, the stop loss and the signal offset $\mu$. We prefer to do a combined optimization of all the parameters. We use daily data of the S\&P 500 index futures (whose CQG code is \textit{EP}) from 01Jan2017 to 01Jan2018. We train our model on the first 6 months and test it on the next six months. Deliberately, our algorithm is unsophisticated to keep thing simple and concentrate on the parameter estimation method. The overall idea is for a given set of parameter to compute the resulting sharpe ratio over the train period and find the optimal parameters combination. For a model like the fourth one in the table \ref{KFModels}, the optimization encompasses 18 parameters: $p_1, \ldots, p_{15}$, the profit target, the stop loss and the signal offset $\mu$, making it non trivial. We use the CMA-ES algorithm to find the optimal solution. In our optimization, we add some penalty condition to force non meaningful Kalman filter parameters to be zero, namely, we add a L1 penalty on this parameters.

Results are given below

\hspace{-1.5cm}
\begin{minipage}{1.2\textwidth}

\begin{table}[H]
  \centering
  \caption{Optimal parameters}
\resizebox{\textwidth}{!}{
    \begin{tabular}{|c|r|r|r|r|r|r|r|r|r|r|r|r|r|r|r|r|r|r|}
    \toprule
    \newline{}\newline{}\newline{}\newline{}\newline{}Parameters & \multicolumn{1}{l|}{$p_{01}$} & \multicolumn{1}{l|}{$p_{02}$} & \multicolumn{1}{l|}{$p_{03}$} & \multicolumn{1}{l|}{$p_{04}$} & \multicolumn{1}{l|}{$p_{05}$} & \multicolumn{1}{l|}{$p_{06}$} & \multicolumn{1}{l|}{$p_{07}$} & \multicolumn{1}{l|}{$p_{08}$} & \multicolumn{1}{l|}{$p_{09}$} & \multicolumn{1}{l|}{$p_{10}$} & \multicolumn{1}{l|}{$p_{11}$} & \multicolumn{1}{l|}{$p_{12}$} & \multicolumn{1}{l|}{$p_{13}$} & \multicolumn{1}{l|}{$p_{14}$} & \multicolumn{1}{l|}{$p_{15}$} & \multicolumn{1}{l|}{offset} & \multicolumn{1}{l|}{stop} & \multicolumn{1}{l|}{target} \\
    \midrule
    Value & 24.8  & 0     & 11.8  & 46.2  & 77.5  & 67    & 100   & 0     & 0     & 0     & 0     & 100   & 0     & 0     & 0     & 5     & 80    & 150 \\
    \bottomrule
    \end{tabular}}%
  \label{tab:param}%
\end{table}%

\begin{table}[H]
  \centering
  \caption{Train test statistics 1/4}
\resizebox{\textwidth}{!}{
    \begin{tabular}{|c|c|c|c|c|c|c|c|c|}
    \toprule
    Performance & Net Profit & Gross Profit & Gross Loss & \# of Trades & \# of Contracts & Avg. Trade & Tot. Net Profit (\%) & Ann. Net Profit (\%) \\
    \midrule
    Train & \textcolor[rgb]{ 0,  .502,  0}{5,086 \euro} & \textcolor[rgb]{ 0,  .502,  0}{11,845 \euro} & \textcolor[rgb]{ 1,  0,  0}{-6,759 \euro} & 15    & 15    & \textcolor[rgb]{ 0,  .502,  0}{339.05 \euro} & \textcolor[rgb]{ 0,  .502,  0}{5.09\%} & \textcolor[rgb]{ 0,  .502,  0}{10.59\%} \\
    \midrule
    Test  & \textcolor[rgb]{ 0,  .502,  0}{4,266 \euro} & \textcolor[rgb]{ 0,  .502,  0}{11,122 \euro} & \textcolor[rgb]{ 1,  0,  0}{-6,857 \euro} & 15    & 15    & \textcolor[rgb]{ 0,  .502,  0}{284.38 \euro} & \textcolor[rgb]{ 0,  .502,  0}{4.27\%} & \textcolor[rgb]{ 0,  .502,  0}{8.69\%} \\
    \bottomrule
    \end{tabular}}
  \label{tab:stat1}%
\end{table}%

\begin{table}[H]
  \centering
  \caption{Train test statistics 2/4}
\resizebox{\textwidth}{!}{
    \begin{tabular}{|c|c|c|c|c|c|c|c|c|}
    \toprule
    Performance & Vol   & Sharpe Ratio & Trades per Day & Avg. Time in Market & Max. Drawdown & Recovery Factor & Daily Ann. Vol & Monthly Ann. Vol \\
    \midrule
    Train & \textcolor[rgb]{ 0,  .502,  0}{6.54\%} & \textcolor[rgb]{ 0,  .502,  0}{1.62} & 0.10  & 8d14h & \textcolor[rgb]{ 1,  0,  0}{-2,941 \euro} & \textcolor[rgb]{ 0,  .502,  0}{3.510} & \textcolor[rgb]{ 0,  .502,  0}{6.54\%} & \textcolor[rgb]{ 0,  .502,  0}{5.72\%} \\
    \midrule
    Test  & \textcolor[rgb]{ 0,  .502,  0}{6.20\%} & \textcolor[rgb]{ 0,  .502,  0}{1.40} & 0.10  & 8d19h & \textcolor[rgb]{ 1,  0,  0}{-1,721 \euro} & \textcolor[rgb]{ 0,  .502,  0}{4.948} & \textcolor[rgb]{ 0,  .502,  0}{6.20\%} & \textcolor[rgb]{ 0,  .502,  0}{5.32\%} \\
    \bottomrule
    \end{tabular}}%
  \label{tab:stat2}%
\end{table}%

\begin{table}[H]
  \centering
  \caption{Train test statistics 3/4}
\resizebox{\textwidth}{!}{
    \begin{tabular}{|c|c|c|c|c|c|c|c|c|}
    \toprule
    Performance & Daily Sharpe Ratio & Daily Sortino Ratio & Commission & Percent Profitable & Profit Factor & \# of Winning Trades & Avg. Winning Trade & Max. conseq. Winners \\
    \midrule
    Train & \textcolor[rgb]{ 0,  .502,  0}{1.62} & \textcolor[rgb]{ 0,  .502,  0}{2.35} & \textcolor[rgb]{ 1,  0,  0}{49 \euro} & 46.67\% & 1.75 \euro & 7     & \textcolor[rgb]{ 0,  .502,  0}{1,692.09 \euro} & 3 \\
    \midrule
    Test  & \textcolor[rgb]{ 0,  .502,  0}{1.40} & \textcolor[rgb]{ 0,  .502,  0}{2.05} & \textcolor[rgb]{ 1,  0,  0}{46 \euro} & 46.67\% & 1.62 \euro & 7     & \textcolor[rgb]{ 0,  .502,  0}{1,588.92 \euro} & 2 \\
    \bottomrule
    \end{tabular}}%
  \label{tab:stat3}%
\end{table}%

\begin{table}[H]
  \centering
  \caption{Train test statistics 4/4}
\resizebox{\textwidth}{!}{
    \begin{tabular}{|c|c|c|c|c|c|c|c|c|}
    \toprule
    Performance & Largest Winning Trade & \# of Losing Trades & Avg. Losing Trade & Max. conseq. Losers & Largest Losing Trade & Avg. Win/Avg. Loss & Avg. Bars in Trade & Time to Recover \\
    \midrule
    Train & \textcolor[rgb]{ 0,  .502,  0}{1,776.11 \euro} & 8     & \textcolor[rgb]{ 1,  0,  0}{-844.85 \euro} & 3     & \textcolor[rgb]{ 1,  0,  0}{-1,011.82 \euro} & 2.00  & 6.1   & 77.00 days \\
    \midrule
    Test  & \textcolor[rgb]{ 0,  .502,  0}{1,609.32 \euro} & 8     & \textcolor[rgb]{ 1,  0,  0}{-857.1 \euro} & 2     & \textcolor[rgb]{ 1,  0,  0}{-860.26 \euro} & 1.85  & 6.2   & 70.00 days \\
    \bottomrule
    \end{tabular}}%
  \label{tab:stat4}%
\end{table}%

\end{minipage}

\newpage

\begin{center}
\captionof{figure}{Kalman filter algorithm on train data set}
\includegraphics[width=13cm]{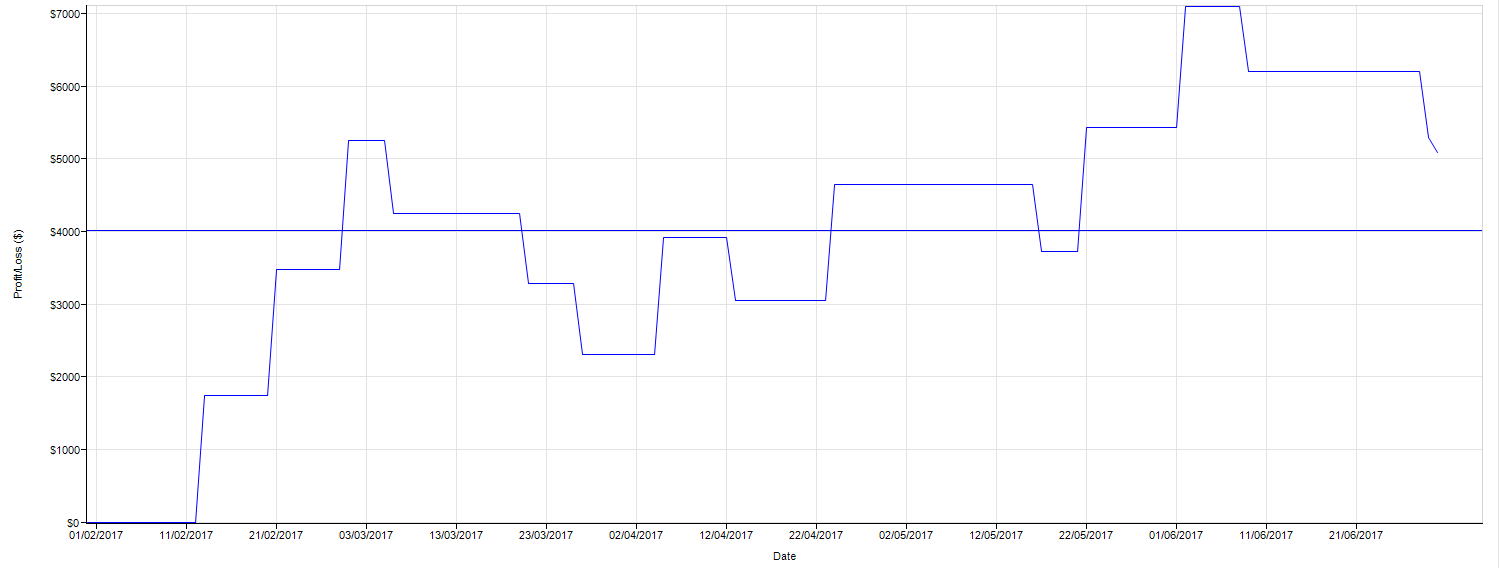}
\label{fig:train}
\end{center}

\begin{center}
\captionof{figure}{Kalman filter algorithm on test data set}
\includegraphics[width=13cm]{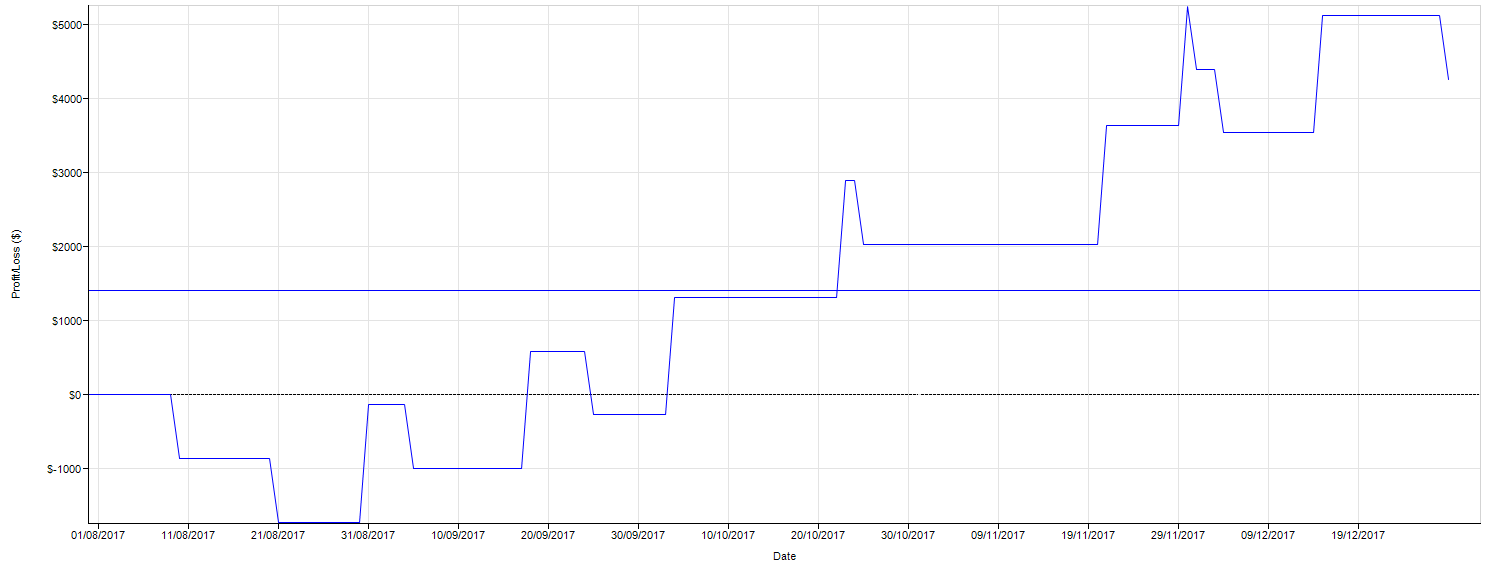}
\label{fig:test}
\end{center}

We compare our algorithm with a traditional moving average crossover algorithm to test the efficiency of Kalman filter for trend detection. The moving average cross over algorithm generates a buy signal when the fast moving average crosses over the long moving average and a sell signal when the former crosses below the latter. A $d$ period moving average is defined as the arithmetic average of the daily close over a $d$ period, denoted by $\mathrm{SMA}(d)$. Our algorithm is given by the following pseudo code

\begin{algorithm}[H]
\caption{Moving Average Trend following algorithm}
	\begin{algorithmic} 
	\State \textbf{Initialize common trade details}
	\State SetProfitTarget( target)							\Comment{fixed profit target in ticks}
	\State SetStopLoss( stop\_loss )							\Comment{fixed stop loss in ticks}
	\\
	\While{ Not In Position}											\Comment{look for new trade}
		\If{ $\mathrm{SMA}\text{(Short)}[0] > \mathrm{SMA}\text{(Long)}[0] + \text{offset}$ } 	\Comment{up trend signal}
			\State EnterLong()										\Comment{market order for the open}
		\ElsIf{  $\mathrm{SMA}\text{(Short)}[0] < \mathrm{SMA}\text{(Long)}[0] + \text{offset}$ } 	\Comment{down trend signal}
			\State EnterShort()										\Comment{market order for the open}
		\EndIf
	\EndWhile
	\end{algorithmic}
\end{algorithm}

We can now compare moving average cross over versus Kalman filter algorith. The table \ref{tab:mavskf} compares the two algorithms. We can see that on the train period, the two algorithms have similar performances : $ 5,260$ vs $5,086$. However on the test period, moving average performs very badly with a net profit of $935$ versus $4,266$ for the bayesian graphical model (the kalman filter) algorithm.

\hspace{-1.5cm}
\begin{minipage}{1.2\textwidth}

\begin{table}[H]
  \centering
  \caption{Moving average cross over versus Kalman filter}
\resizebox{\textwidth}{!}{
    \begin{tabular}{|c|c|c|c|c|c|c|c|c|}
    \toprule
    Algo  & Total Net Profit & Recovery Factor & Profit Factor & Max. Drawdown & Sharpe Ratio & Total \# of Trades & Percent Profitable & Train: Total Net Profit \\
    \midrule
    \textcolor[rgb]{ .125,  .122,  .208}{MA Cross over} & \textcolor[rgb]{ 0,  .502,  0}{935 \euro} &                    0.32  &             1.13  & \textcolor[rgb]{ 1,  0,  0}{-\euro 2,889} &              0.41  &                         26  &                      0.54  &                     5,260.00  \\
    \midrule
    \textcolor[rgb]{ .125,  .122,  .208}{Kalman filter} & \textcolor[rgb]{ 0,  .502,  0}{4,266 \euro} &                    2.48  &             1.62  & \textcolor[rgb]{ 1,  0,  0}{-\euro 1,721} &              1.40  &                         30  &                      0.47  &                     5,085.79  \\
    \bottomrule
    \end{tabular}}%
  \label{tab:mavskf}%
\end{table}
\end{minipage}

\begin{center}
\captionof{figure}{Moving Average Crossover algorithm on train data set}
\includegraphics[width=13cm]{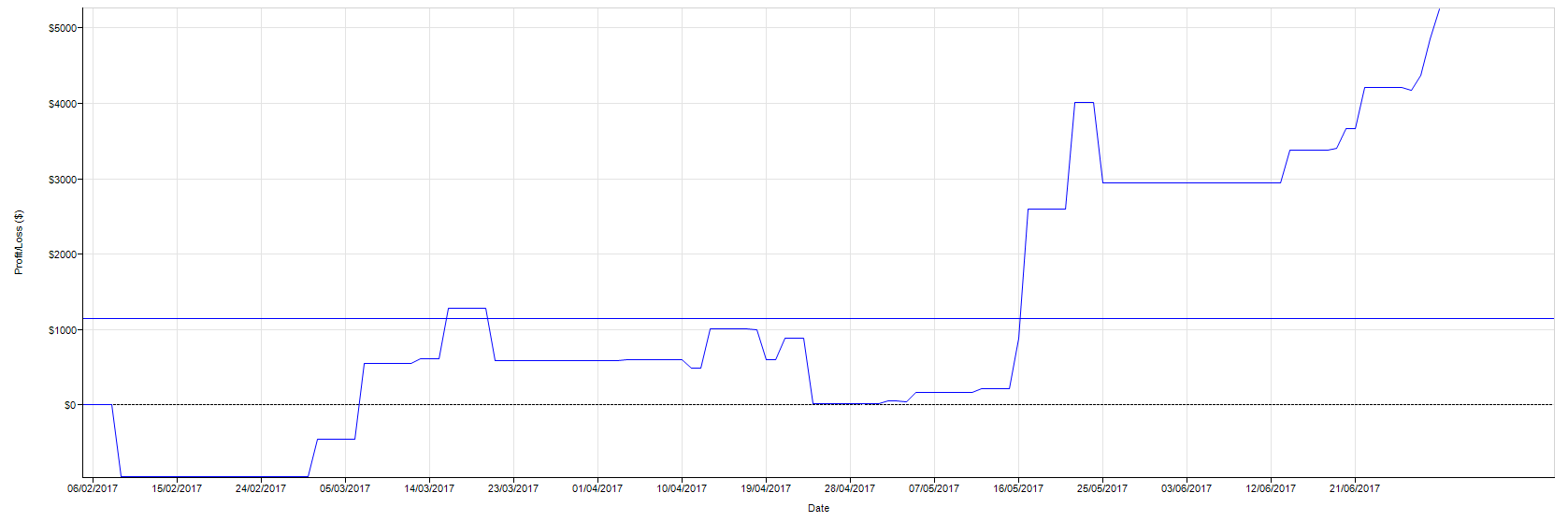}
\label{fig:ma_train}
\end{center}

\newpage
\begin{center}
\captionof{figure}{Moving Average Crossover algorithm on test data set}
\includegraphics[width=13cm]{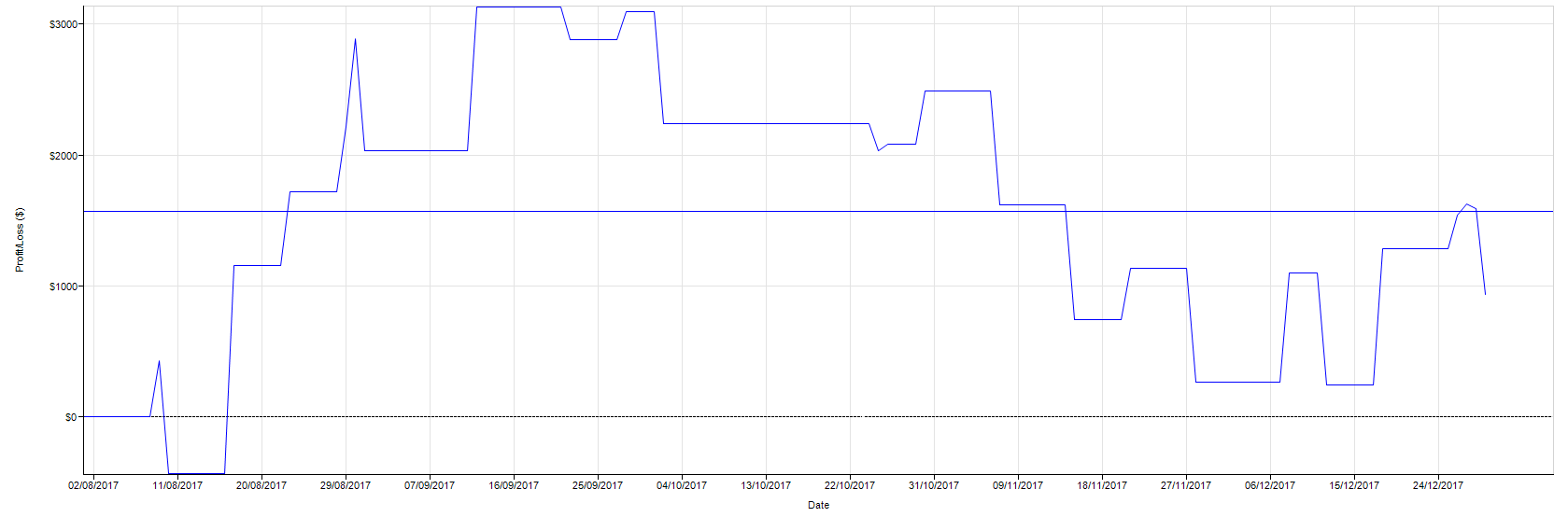}
\label{fig:ma_test}
\end{center}

\section{Conclusion}
In this paper, we have revisited the Kalman filter theory and showed the strong connection between Kalman filter and Hidden Markov Models. After discussing various recursive scheme for the underlying Bayesian graphical model, we provided various insight on the parameter estimation. The traditional method is to use the expectation maximization (EM) method. This method aims to find the best fit distribution according to our modeling assumptions. This may work well when the state space model is ruled by physical law. However, in finance, there is no such a thing as the physical law of markets. One makes assumptions about the price dynamics that may be completely wrong. Hence, we present a new method based on CMA-ES optimization that does not decouple the Bayesian graphical model parameters from the other parameters of the financial strategy. We show on a numerical example on the S\&P 500 index futures over one year of data (from Jan 1st 2017 to Jan 1st 2018) that the combined estimation of the Bayesian graphical model works well and that there is no much performance deterioration between train and test. This suggests little over fitting. We also show that the traditional moving average cross over method does not perform that well compared to our method and show strong over fitting bias as the net performance over the test data set is dramatically lower than the one on the train, suggesting over fitting.

\clearpage

\appendix

\section{Kalman filter proof}\label{Kalman_filter_proof}
The proof is widely known and presented in various books (see for instance \cite{Gelb_1986} or \cite{Brown_1997}). We present it here for the sake of completeness. It consists in three steps:
\begin{itemize}
\item Derive the posteriori estimate covariance matrix
\item Compute the Kalman gain
\item Simplify the posteriori error covariance formula
\end{itemize}

\subsection{Step 1: A posteriori estimate covariance matrix} \label{step1}
\begin{proof}
The error covariance $\mathbf{P}_{t\mid t}$ is defined as the covariance between $\mathbf{x}_t$ and $\mathbf{x}_{t+1}$: $\mathbf{P}_{t\mid t} = \operatorname{Cov}\left(\mathbf{x}_t - \hat{\mathbf{x}}_{t\mid t}\right)$. 
We can substitute in this equation the definition of $\hat{\mathbf{x}}_{t\mid t}$  given by equation (\ref{correct_eq4}) to get:
\begin{equation}
\mathbf{P}_{t\mid t} = \operatorname{Cov}\left[\mathbf{x}_t - \left(\hat{\mathbf{x}}_{t\mid t-1} + \mathbf{K}_t\tilde{\mathbf{y}}_t\right)\right]
\end{equation}

Using the definition of $\tilde{\mathbf{y}}_t$ from equation (\ref{correct_eq1}), we get:

\begin{equation}
\mathbf{P}_{t\mid t} = \operatorname{Cov}\left(\mathbf{x}_t - \left[\hat{\mathbf{x}}_{t\mid t-1} + \mathbf{K}_t\left(\mathbf{z}_t - \mathbf{H}_t\hat{\mathbf{x}}_{t\mid t-1}\right)\right]\right)
\end{equation}

Using the measurement equation for $\mathbf{z}_t$ (equation  (\ref{measurement_equation})), this transforms into:
\begin{equation}
\mathbf{P}_{t\mid t} = \operatorname{Cov}\left(\mathbf{x}_t - \left[\hat{\mathbf{x}}_{t\mid t-1} + \mathbf{K}_t\left(\mathbf{H}_t\mathbf{x}_t + \mathbf{v}_t - \mathbf{H}_t\hat{\mathbf{x}}_{t\mid t-1}\right)\right]\right)
\end{equation}

By collecting the error vectors, this results in:
\begin{equation}
\mathbf{P}_{t\mid t} = \operatorname{Cov}\left[\left(\mathbf{I} - \mathbf{K}_t \mathbf{H}_t)(\mathbf{x}_t - \hat{\mathbf{x}}_{t\mid t-1}\right) - \mathbf{K}_t \mathbf{v}_t \right]
\end{equation}

Since the measurement error $v_{k}$ is not correlated with the other terms, this simplifies into:
\begin{equation}
\mathbf{P}_{t\mid t} = \operatorname{Cov}\left[\left(\mathbf{I} - \mathbf{K}_t \mathbf{H}_{k}\right)\left(\mathbf{x}_t - \hat{\mathbf{x}}_{t\mid t-1}\right)\right] + \operatorname{Cov}\left[\mathbf{K}_t \mathbf{v}_t\right]
\end{equation}

By property of covariance for matrix: $\operatorname{Cov}(A \mathbf{x}_t ) = A \operatorname{Cov}(\mathbf{x}_t ) A^\mathrm{T}$ and using the definition for 
$\mathbf{P}_{t\mid t-1}$ and $\mathbf{R}_t = \operatorname{Cov}\left(\mathbf{v}_t\right)$, we get the final result:
\begin{equation}\label{Joseph_form}
\mathbf{P}_{t\mid t} = (\mathbf{I} - \mathbf{K}_t \mathbf{H}_t) \mathbf{P}_{t\mid t-1} (\mathbf{I} - \mathbf{K}_t \mathbf{H}_t)^\mathrm{T} + \mathbf{K}_t \mathbf{R}_t \mathbf{K}_t^\mathrm{T} 
\end{equation}

This last equation is referred to as the \textit{Joseph form} of the covariance update equation. It is true for any value of Kalman gain $\mathbf{K}_t$ no matter if it is optimal or not. In the special case of optimal Kalman gain, this further reduces to equation (\ref{reduced_form}) which we will prove in \ref{step3}
\end{proof}

\subsection{Step 2: Kalman gain}\label{step2}
\begin{proof}
The goal of the Kalman filter is to find the minimum mean-square error estimator. It is obtained by minimizing the error of the \textit{a posteriori} state given by 
$\mathbf{x}_t - \hat{\mathbf{x}}_{t\mid t}$. This error is stochastic, hence we are left with minimizing the expected value of the square of the $L_2$ norm 
of this vector given by $\operatorname{E}\left[\left\|\mathbf{x}_{t} - \hat{\mathbf{x}}_{t|t}\right\|^2\right]$. 
Since the error of the \textit{a posteriori} state follows a normal distribution with zero mean and covariance given by $\mathbf{P}_{t|t}$, the optimization program is equivalent to minimizing the matrix trace of the a posteriori estimate covariance matrix $\mathbf{P}_{t|t}$. \\

Expanding out and collecting the terms in equation (\ref{Joseph_form}), we find:
\begin{align}\label{expandedequation}
 \mathbf{P}_{t\mid t} & = \mathbf{P}_{t\mid t-1} - \mathbf{K}_t \mathbf{H}_t \mathbf{P}_{t\mid t-1} - \mathbf{P}_{t\mid t-1} \mathbf{H}_t^\mathrm{T} \mathbf{K}_t^\mathrm{T} + \mathbf{K}_t \mathbf{S}_t\mathbf{K}_t^\mathrm{T}
\end{align}

This minimization program is very simple. It is quadratic in the Kalman gain, $\mathbf{K}_t$. The optimum is obtained when the derivative with respect to $\mathbf{K}_t$ is null. The equation for the critical point writes as:
\begin{equation}
\frac{\partial \; \mathrm{tr}(\mathbf{P}_{t\mid t})}{\partial \;\mathbf{K}_t} = -2 (\mathbf{H}_t \mathbf{P}_{t\mid t-1})^\mathrm{T} + 2 \mathbf{K}_t \mathbf{S}_t = 0.
\end{equation}

We can solving for $\mathbf{K}_t$ to find the \textit{optimal} Kalman gain as follows:

\begin{align} \label{KalmanGain}
\mathbf{K}_t &= \mathbf{P}_{t\mid t-1} \mathbf{H}_t^\mathrm{T} \mathbf{S}_t^{-1}
\end{align}
\end{proof}

\subsection{Step 3: Simplification of the posteriori error covariance formula}\label{step3}
\begin{proof}
We can work on our formula further to simplify equation (\ref{Joseph_form}) as follows. A trick is to remark that in equation (\ref{KalmanGain}), multiplying on the right both sides by $\mathbf{S}_t \mathbf{K}_t^\mathrm{T} $, we find that:
\begin{equation}
\mathbf{K}_t \mathbf{S}_t \mathbf{K}_t^\mathrm{T} = \mathbf{P}_{t\mid t-1} \mathbf{H}_t^\mathrm{T} \mathbf{K}_t^\mathrm{T}
\end{equation}

Injecting this equality in equation (\ref{expandedequation}) and noticing that the last two terms cancel out, we get:
\begin{equation}\label{simpleequation}
\mathbf{P}_{t\mid t} =  (\mathbf{I} - \mathbf{K}_t \mathbf{H}_t) \mathbf{P}_{t\mid t-1}.
\end{equation}

This last formula appeals some remarks. It makes Kalman filter dammed simple. It is computationally  cheap and much cheaper than the equation (\ref{expandedequation}). Thus it is nearly always used in practice. However, this equation is only correct for the optimal gain. If by back luck, arithmetic precision is unusually low due to numerical instability, or if the Kalman filter gain is non-optimal (intentionally or not), equation \ref{simpleequation} does not hold any more. Using it would lead to inaccurate results and can be troublesome. In this particular case, the complete formula given by equation (\ref{expandedequation}) must be used.
\end{proof}

\section{Factor Analysis model}\label{factor_analysis_proof}
\subsection{Model Presentation and proof}
We are interested in the following very simple graphical model called the Gaussian Factor analysis model. It is represented by a two nodes graphical models in figure \ref{factor_analysis}

\begin{figure}[h]
\centering

\begin{tikzpicture}
 \node[main, label= {$X$}] (L1) []  {};
 \node[main,fill=black!35,  label=below:{$Z$}] (O1) [below=of L1] {}; 
 \path (L1) edge [connect] (O1);
\end{tikzpicture}
\caption{Factor analysis model as a graphical model} \label{factor_analysis}
\end{figure}
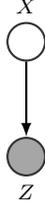

We assume that $X$ follows a multi dimensional normal distribution $X \sim \mathcal{N}(0, \Sigma)$. We assume that $X$  is a latent variable and that only $Z$ is observed. We also assume that there is a direct linear relationship between $X$ and $Z$ given by

$$
Z = \mu + \Lambda X + W
$$
where $W$ is also distributed as a a multi dimensional normal distribution $W \sim \mathcal{N}(0, \Psi)$ independent of $X$. This is obviously a simplified version of our Linear Gaussian State Space Model and can provide building blocks for deriving the Kalman filter. Trivially, we can compute the unconditional distribution of $Z$ as a normal distribution whose mean is $\mathbb{E}[Z]=\mathbb{E}[ \mu + \Lambda X + W] = \mu$. Likewise, it is immediate to compute the unconditional covariance matrix of $Z$ as 
\begin{eqnarray}
\operatorname{Var}(Z) &=& \mathbb{E}[ (\mu + \Lambda X + W)(\mu + \Lambda X + W)^\mathrm{T} ] \\
&=& \Lambda \Sigma \Lambda^\mathrm{T} + \Psi
\end{eqnarray}

The covariance between $X$ and $Z$ is also easy to compute and given by

\begin{eqnarray}
\operatorname{Cov}(X,Z) &=& \mathbb{E}[ X(\mu + \Lambda X + W)^\mathrm{T} ] \\
&=& \Sigma \Lambda^\mathrm{T}
\end{eqnarray}

Hence, collecting all our results, we have shown that the joint distribution of $X$ and $Z$ is a Gaussian given by
$$
 \mathcal{N}\left( \left[ \begin{array}{c} 0 \\ \mu \end{array} \right], 
\left[
\begin{array}{l l }
\Sigma 			&  \Sigma \Lambda^\mathrm{T} \\
\Lambda  \Sigma 	&  \Lambda \Sigma \Lambda^\mathrm{T} + \Psi
\end{array}
\right] \right).
$$

We shall be interested in computing the conditional distribution of $X$  given $Z$. Using lemma provided in \ref{conditionalGaussian}, we obtained that $X \mid Z$ is a multi dimensional Gaussian whose mean is:
\begin{eqnarray}
\mathbb{E}(X \mid z ) &= & \Sigma \Lambda^\mathrm{T} (\Lambda \Sigma \Lambda^\mathrm{T} + \Psi)^{-1} (z-\mu)
\end{eqnarray}

Using the Woodbury matrix inversion formula (see formula (158) in \cite{Petersen_2012}), this is also
\begin{eqnarray}
\mathbb{E}(X \mid z ) &= &(\Sigma^{-1} + \Lambda^\mathrm{T} \Psi^{-1} \Lambda^{-1})  \Lambda^\mathrm{T} \Psi^{-1} (z-\mu)
\end{eqnarray}

 The two formula are equivalent and one should use the one that is easier to compute depending on the dimension of $X$ and $Z$. Our partitioned Gaussian lemma \ref{conditionalGaussian} gives us also the conditional variance of $X$:

\begin{eqnarray}
\operatorname{Var}(X \mid z) &=& \Sigma - \Sigma \Lambda^\mathrm{T} ( \Lambda \Sigma \Lambda^\mathrm{T} + \Psi)^{-1} \Lambda  \Sigma \\
&=& (\Sigma^{-1}  + \Lambda^\mathrm{T} \Psi^{-1} \Lambda)^{-1}
\end{eqnarray}

where in the last equation, we have used again the Woodbury matrix inversion formula (Woodbury variant formula (157) in \cite{Petersen_2012}).

\subsection{A quick lemma about partitioned Gaussian}\label{conditionalGaussian}
\begin{lemma}
Let $Y$ be a multivariate normal vector ${\boldsymbol Y} \sim \mathcal{N}(\boldsymbol\mu, \Sigma)$. Consider partitioning ${\boldsymbol Y}$, $\boldsymbol\mu$, $\Sigma$ into 

$$
\boldsymbol\mu =
\begin{bmatrix}
 \boldsymbol\mu_1 \\
 \boldsymbol\mu_2
\end{bmatrix} 
\qquad \qquad 
{\boldsymbol Y}=
\begin{bmatrix}
{\boldsymbol y}_1 \\ 
{\boldsymbol y}_2 
\end{bmatrix}
\qquad \qquad 
\Sigma=
\begin{bmatrix}
\Sigma_{11} & \Sigma_{12}\\
\Sigma_{21} &\Sigma_{22}
\end{bmatrix}
$$

Then, $({\boldsymbol y}_1|{\boldsymbol y}_2={\boldsymbol a})$, the conditional distribution of the first partition given the second, is 
$\mathcal{N}(\overline{\boldsymbol\mu},\overline{\Sigma})$, with mean
$$
\boldsymbol\mu_{1\mid2}=\boldsymbol\mu_1+\Sigma_{12}{\Sigma_{22}}^{-1}({\boldsymbol a}-\boldsymbol\mu_2)
$$
and covariance matrix
$$
 \Sigma_{1\mid2}=\Sigma_{11}-\Sigma_{12}{\Sigma_{22}}^{-1}\Sigma_{21}$$
\end{lemma}

\begin{proof}
Let ${\bf y}_{1}$ be the first partition and ${\bf y}_2$ the second. Now define ${\bf z} = {\bf y}_1 + {\bf A} {\bf y}_2 $ where ${\bf A} = -\Sigma_{12} \Sigma^{-1}_{22}$. We can easily compute the covariance between $z$ and $y$ as follows:

\begin{align*} {\rm cov}({\bf z}, {\bf y}_2) &= {\rm cov}( {\bf y}_{1}, {\bf y}_2 ) + 
{\rm cov}({\bf A}{\bf y}_2, {\bf y}_2) \\
&= \Sigma_{12} + {\bf A} \operatorname{Var}({\bf y}_2) \\
&= 0
\end{align*}

Since ${\bf z}$ and ${\bf y}_2$ are jointly normal and uncorrelated, they are independent. 
The expectation  of ${\bf z}$ is trivially computed as 
$E({\bf z}) = {\boldsymbol \mu}_1 + {\bf A}  {\boldsymbol \mu}_2$. 
Using this, we compute the conditional expectation of ${\bf y}_1$ given ${\bf y}_2$ as follows:

\begin{align*}
E({\bf y}_1 | {\bf y}_2) &= E( {\bf z} - {\bf A} {\bf y}_2 | {\bf y}_2) \\
& = E({\bf z}|{\bf y}_2) -  E({\bf A}{\bf y}_2|{\bf y}_2) \\
& = E({\bf z}) - {\bf A}{\bf y}_2 \\
& = {\boldsymbol \mu}_1 + \Sigma_{12} \Sigma^{-1}_{22} ({\bf y}_2- {\boldsymbol \mu}_2)
\end{align*}

which proves the first part. The second part is as simple. Note that 

\begin{align*}
\operatorname{Var}({\bf y}_1|{\bf y}_2) &= \operatorname{Var}({\bf z} - {\bf A} {\bf y}_2 | {\bf y}_2) \\
&= \operatorname{Var}({\bf z}|{\bf y}_2) \\
&= \operatorname{Var}({\bf z})
\end{align*}
since ${\bf z}$ and ${\bf y}_2$ are independent. We can trivially conclude using the following last computation:
\begin{align*}
\operatorname{Var}({\bf y}_1|{\bf y}_2) = \operatorname{Var}( {\bf z} ) &= \operatorname{Var}( {\bf y}_1 + {\bf A} {\bf y}_2 ) \\
&= \operatorname{Var}( {\bf y}_1 ) + {\bf A} \operatorname{Var}( {\bf y}_2 ) {\bf A}^\mathrm{T}
+ {\bf A} {\rm cov}({\bf y}_1,{\bf y}_2) + {\rm cov}({\bf y}_2,{\bf y}_1) {\bf A}^\mathrm{T} \\
&= \Sigma_{11} +\Sigma_{12} \Sigma^{-1}_{22} \Sigma_{22}\Sigma^{-1}_{22}\Sigma_{21}
- 2 \Sigma_{12} \Sigma_{22}^{-1} \Sigma_{21} \\
&= \Sigma_{11} -\Sigma_{12} \Sigma^{-1}_{22}\Sigma_{21}
\end{align*}
\end{proof}

\subsection{Derivation of the filter equations}\label{filter_equation_proof}
\begin{proof}
It is easy to start with the precision matrix (the inverse of the covariance matrix). We have
\begin{eqnarray}
\mathbf{\Lambda}_{t \mid t-1} &=&  \mathbf{P}_{t \mid t-1}^{-1} \\
&=& (\mathbf{F}_{t} \mathbf{P}_{t-1 \mid t-1} \mathbf{F}_{t}^\mathrm{T} + \mathbf{Q}_t)^{-1} \\
&=& \mathbf{Q}_t^{-1}-\mathbf{Q}_t^{-1} \mathbf{F}_{t} (\mathbf{P}_{t-1 \mid t-1}^{-1} + \mathbf{F}_{t}^\mathrm{T} \mathbf{Q}_t^{-1} \mathbf{F}_{t})^{-1}  \mathbf{F}_{t}^\mathrm{T} \mathbf{Q}_t^{-1} \\
&=& \mathbf{Q}_t^{-1}-\mathbf{Q}_t^{-1} \mathbf{F}_{t} (\mathbf{\Lambda}_{t-1 \mid t-1} + \mathbf{F}_{t}^\mathrm{T} \mathbf{Q}_t^{-1} \mathbf{F}_{t})^{-1}  \mathbf{F}_{t}^\mathrm{T} \mathbf{Q}_t^{-1} 
\end{eqnarray}

where in the previous equation, we have used extensively matrix inversion formula. 
Similarly, applying matrix inversion but for $\mathbf{\Lambda}_{t \mid t}$ (conditioned on $t$ now) 
\begin{eqnarray}
\mathbf{\Lambda}_{t \mid t} &=&  \mathbf{P}_{t \mid t}^{-1} \\
&=&(  \mathbf{P}_{t \mid t-1}- \mathbf{P}_{t \mid t-1}  \mathbf{H}_{t}^\mathrm{T} (\mathbf{H}_{t}  \mathbf{P}_{t \mid t-1}  \mathbf{H}_{t}^\mathrm{T} + \mathbf{R}_{t})^{-1} \mathbf{H}_{t}  \mathbf{P}_{t \mid t-1} )^{-1} \\
&=&\mathbf{P}_{t \mid t-1}^{-1}+\mathbf{H}_{t}^\mathrm{T} \mathbf{R}_{t}^{-1}\mathbf{H}_{t}\\
&=&\mathbf{\Lambda}_{t \mid t-1} +\mathbf{H}_{t}^\mathrm{T} \mathbf{R}_{t}^{-1}\mathbf{H}_{t}
\end{eqnarray}

We can now handle the parameter $\eta$. We have
\begin{eqnarray}
\hat{\eta}_{t \mid t-1} &=& \mathbf{P}_{t \mid t-1}^{-1} \hat{\mathbf{x}}_{t \mid t-1} \\
&=& \mathbf{P}_{t \mid t-1}^{-1}  \left( \mathbf{F}_{t} \hat{\mathbf{x}}_{t-1 \mid t-1} + \mathbf{B}_t   \mathbf{u}_t \right) \\
&=&\mathbf{P}_{t \mid t-1}^{-1}   \left( \mathbf{F}_{t} \mathbf{P}_{t-1 \mid t-1} \hat{\eta}_{t-1 \mid t-1} + \mathbf{B}_t   \mathbf{u}_t \right)\\
&=&( \mathbf{F}_{t} \mathbf{P}_{t-1 \mid t-1} \mathbf{F}_{t}^\mathrm{T} + \mathbf{Q}_{t})^{-1} \left(\mathbf{F}_{t} \mathbf{P}_{t-1 \mid t-1} \hat{\eta}_{t-1 \mid t-1} + \mathbf{B}_t   \mathbf{u}_t \right) \nonumber \\
&=& \mathbf{Q}_{t}^{-1} \mathbf{F}_{t} (\mathbf{P}_{t-1 \mid t-1}^{-1}+\mathbf{F}_{t}^\mathrm{T}  \mathbf{Q}_{t}^{-1} \mathbf{F}_{t})^{-1}\hat{\eta}_{t-1 \mid t-1} \\
& & \hspace{3cm}+ ( \mathbf{F}_{t} \mathbf{P}_{t-1 \mid t-1} \mathbf{F}_{t}^\mathrm{T} + \mathbf{Q}_{t})^{-1} \mathbf{B}_t   \mathbf{u}_t  \\
&=& \mathbf{Q}_{t}^{-1} \mathbf{F}_{t} (\mathbf{\Lambda}_{t-1 \mid t-1} +\mathbf{F}_{t}^\mathrm{T}  \mathbf{Q}_{t}^{-1} \mathbf{F}_{t})^{-1}\hat{\eta}_{t-1 \mid t-1} \nonumber \\
& & \!\!\!\!\!  + (\mathbf{Q}_{t}^{-1} \! \!- \! \mathbf{Q}_{t}^{-1} \mathbf{F}_{t} ( \mathbf{\Lambda}_{t-1 \mid t-1} \!+ \! \mathbf{F}_{t}^\mathrm{T} \mathbf{Q}_{t}^{-1} \mathbf{F}_{t})^{-1} \mathbf{F}_{t}^\mathrm{T} \mathbf{Q}_{t}^{-1} ) \mathbf{B}_t  \mathbf{u}_t 
\end{eqnarray}

Likewise, we derive the same type of equations for $\eta$ but conditioned on $t$ as follows:
\begin{eqnarray}
\hat{\eta}_{t \mid t} &=& \mathbf{P}_{t \mid t}^{-1} \hat{\mathbf{x}}_{t \mid t}  \\
&=& \mathbf{P}_{t \mid t}^{-1} (\hat{\mathbf{x}}_{t \mid t-1} + \mathbf{P}_{t \mid t}\mathbf{H}_{t}^\mathrm{T}\mathbf{R}_{t}^{-1}(\mathbf{z}_{t}-\mathbf{H}_{t}\hat{\mathbf{x}}_{t \mid t-1}))\\
&=& (\mathbf{P}_{t \mid t}^{-1} -\mathbf{H}_{t}^\mathrm{T}\mathbf{R}_{t}^{-1} \mathbf{H}_{t})\mathbf{P}_{t \mid t-1} \hat{\eta}_{t \mid t-1}+\mathbf{H}_{t}^\mathrm{T}\mathbf{R}_{t}^{-1}\mathbf{z}_{t}\\
&=& (\mathbf{P}_{t \mid t-1}^{-1} +\mathbf{H}_{t}^\mathrm{T}\mathbf{R}_{t}^{-1} \mathbf{H}_{t} -\mathbf{H}_{t}^\mathrm{T}\mathbf{R}_{t}^{-1} \mathbf{H}_{t})\mathbf{P}_{t \mid t-1} \hat{\eta}_{t \mid t-1}  \nonumber \\
& & \hspace{5cm} +\mathbf{H}_{t}^\mathrm{T}\mathbf{R}_{t}^{-1}\mathbf{z}_{t}\\
&=& \hat{\eta}_{t \mid t-1}+\mathbf{H}_{t}^\mathrm{T}\mathbf{R}_{t}^{-1} \mathbf{z}_{t}
\end{eqnarray}

Summarizing all these equation leads to the so called filter equations:
\begin{eqnarray}
\hat{\eta}_{t \mid t-1} &=& \mathbf{Q}_{t}^{-1} \mathbf{F}_{t} (\mathbf{\Lambda}_{t-1 \mid t-1} +\mathbf{F}_{t}^\mathrm{T}  \mathbf{Q}_{t}^{-1} \mathbf{F}_{t})^{-1}\hat{\eta}_{t-1 \mid t-1} \nonumber \\
& &\!\!\!\!\! + (\mathbf{Q}_{t}^{-1} \! \!- \! \mathbf{Q}_{t}^{-1} \mathbf{F}_{t} ( \mathbf{\Lambda}_{t-1 \mid t-1} \!+ \! \mathbf{F}_{t}^\mathrm{T} \mathbf{Q}_{t}^{-1} \mathbf{F}_{t})^{-1} \mathbf{F}_{t}^\mathrm{T} \mathbf{Q}_{t}^{-1} ) \mathbf{B}_t  \mathbf{u}_t \\
\hat{\eta}_{t \mid t}&=&\hat{\eta}_{t \mid t-1}+\mathbf{H}_{t}^\mathrm{T}\mathbf{R}_{t}^{-1} \mathbf{z}_{t}\\
\mathbf{\Lambda}_{t \mid t-1} &=&  \mathbf{Q}_t^{-1}-\mathbf{Q}_t^{-1} \mathbf{F}_{t} (\mathbf{\Lambda}_{t-1 \mid t-1} + \mathbf{F}_{t}^\mathrm{T} \mathbf{Q}_t^{-1} \mathbf{F}_{t})^{-1}  \mathbf{F}_{t}^\mathrm{T} \mathbf{Q}_t^{-1} \\
\mathbf{\Lambda}_{t \mid t} &=&  \mathbf{\Lambda}_{t \mid t-1} +\mathbf{H}_{t}^\mathrm{T} \mathbf{R}_{t}^{-1}\mathbf{H}_{t} 
\end{eqnarray}
which concludes the proof.
\end{proof}

\subsection{Proof of RTS recursive equations} \label{RTS_recursive_equation_proof}
\begin{proof}

We start by writing down the distribution of $\mathbf{x}_{t}$  and $\mathbf{x}_{t+1}$ conditioned on $\mathbf{z}_{1}, \ldots, \mathbf{z}_{t}$. 
As $ \hat{\mathbf{x}}_{t+1 \mid t} =  \mathbf{F}_{t+1} \mathbf{x}_{t} + \mathbf{B}_{t+1}  \mathbf{u}_{t+1}$ and using the fact that the control term $\mathbf{B}_{t+1} \mathbf{u}_{t+1}$ is deterministic, we have

\begin{eqnarray}
\mathbb{E}[(\mathbf{x}_{t}-\hat{\mathbf{x}}_{t\mid t} ) (\mathbf{x}_{t+1}- \hat{\mathbf{x}}_{t+1 \mid t})^\mathrm{T}  \mid \mathbf{z}_{1}, \ldots, \mathbf{z}_{t} ] = \mathbf{P}_{t \mid t}  \mathbf{F}_{t+1} ^\mathrm{T} 
\end{eqnarray}

Thus the vector distribution of $\left[ \begin{array}{c} \mathbf{x}_{t\mid t} \\ \mathbf{x}_{t+1\mid t} \end{array} \right]$ has its mean given by
$\left[ \begin{array}{c} \hat{\mathbf{x}}_{t\mid t} \\ \hat{\mathbf{x}}_{t+1\mid t} \end{array} \right]$
and its covariance matrix  given by
$
\left[ \begin{array}{l l}
\mathbf{P}_{t \mid t} &  \mathbf{P}_{t \mid t}  \mathbf{F}_{t+1} ^\mathrm{T}  \\
 \mathbf{F}_{t+1} \mathbf{P}_{t \mid t}& \mathbf{P}_{t+1 \mid t}
\end{array} \right]
$

Let us now work on the backward recursion. We are interested in computing the distribution of $\mathbf{x}_{t}$, conditioned on  $\mathbf{x}_{t+1}$  and $\mathbf{z}_{1}, \ldots, \mathbf{z}_{t}$. Using the Factor analysis model, it is easy to see that 

\begin{eqnarray}
\mathbb{E}[\mathbf{x}_{t}  \mid \mathbf{x}_{t+1}, \mathbf{z}_{1}, \ldots, \mathbf{z}_{t} ] &=& \hat{\mathbf{x}}_{t\mid t}  +  \mathbf{P}_{t \mid t} \mathbf{F}_{t+1} ^\mathrm{T}
 \mathbf{P}_{t+1 \mid t}^{-1} (\mathbf{x}_{t+1} - \hat{\mathbf{x}}_{t+1\mid t} )
\end{eqnarray}

Similarly, we have
\begin{eqnarray}
\operatorname{Var}\left[ \mathbf{x}_{t}  \mid \mathbf{x}_{t+1}, \mathbf{z}_{1}, \ldots, \mathbf{z}_{t} \right] &=& \mathbf{P}_{t \mid t} - \mathbf{P}_{t \mid t}  \mathbf{F}_{t+1}^\mathrm{T} 
 \mathbf{P}_{t+1 \mid t}^{-1} \mathbf{F}_{t+1}  \mathbf{P}_{t \mid t}
\end{eqnarray}

But using the Markov property that states that 
\begin{eqnarray}
\mathbb{E}[\mathbf{x}_{t}  \mid \mathbf{x}_{t+1}, \mathbf{z}_{1}, \ldots, \mathbf{z}_{T} ] &=& \mathbb{E}[\mathbf{x}_{t}  \mid \mathbf{x}_{t+1}, \mathbf{z}_{1}, \ldots, \mathbf{z}_{t} ]
\end{eqnarray}

and
\begin{eqnarray}
\operatorname{Var}[\mathbf{x}_{t}  \mid \mathbf{x}_{t+1}, \mathbf{z}_{1}, \ldots, \mathbf{z}_{t} ] &=& \operatorname{Var}(\mathbf{x}_{t}  \mid \mathbf{x}_{t+1}, \mathbf{z}_{1}, \ldots, \mathbf{z}_{T} ] 
\end{eqnarray}

We have 

\begin{eqnarray}
\mathbb{E}[\mathbf{x}_{t}  \mid \mathbf{x}_{t+1}, \mathbf{z}_{1}, \ldots, \mathbf{z}_{T} ] &=& \hat{\mathbf{x}}_{t\mid t}  +  \mathbf{P}_{t \mid t} \mathbf{F}_{t+1} ^\mathrm{T}
 \mathbf{P}_{t+1 \mid t}^{-1} (\mathbf{x}_{t+1} - \hat{\mathbf{x}}_{t+1\mid t} )
\end{eqnarray}

Similarly, we have
\begin{eqnarray}
\operatorname{Var}[ \mathbf{x}_{t}  \mid \mathbf{x}_{t+1}, \mathbf{z}_{1}, \ldots, \mathbf{z}_{T} ] &=& \mathbf{P}_{t \mid t} - \mathbf{P}_{t \mid t}  \mathbf{F}_{t+1}^\mathrm{T} 
 \mathbf{P}_{t+1 \mid t}^{-1} \mathbf{F}_{t+1}  \mathbf{P}_{t \mid t}^{-1}
\end{eqnarray}

Using conditional properties states in appendix section \ref{conditional_formula}, we can work the recursion as follows:

\begin{eqnarray}
\hat{\mathbf{x}}_{t\mid T} & \triangleq & \mathbb{E}[\mathbf{x}_{t}  \mid \mathbf{z}_{1}, \ldots, \mathbf{z}_{T} ] \\
&= & \mathbb{E}[ \mathbb{E}[ \mathbf{x}_{t}  \mid \mathbf{x}_{t+1}, \mathbf{z}_{1}, \ldots, \mathbf{z}_{T} ] \mid \mathbf{z}_{1}, \ldots, \mathbf{z}_{T}  ]\\
&= & \mathbb{E}[ \hat{\mathbf{x}}_{t\mid t}  +  \mathbf{P}_{t \mid t} \mathbf{F}_{t+1} ^\mathrm{T}
 \mathbf{P}_{t+1 \mid t}^{-1} (\mathbf{x}_{t+1} - \hat{\mathbf{x}}_{t+1\mid t} ) \mid \mathbf{z}_{1}, \ldots, \mathbf{z}_{T}  ]\\
&= & \hat{\mathbf{x}}_{t\mid t}  +  \mathbf{P}_{t \mid t} \mathbf{F}_{t+1} ^\mathrm{T}
 \mathbf{P}_{t+1 \mid t}^{-1} (\mathbf{x}_{t+1 \mid T} - \hat{\mathbf{x}}_{t+1\mid t} ) \\
& = & \hat{\mathbf{x}}_{t\mid t}  +  \mathbf{L}_{t}(\mathbf{x}_{t+1 \mid T} - \hat{\mathbf{x}}_{t+1\mid t} )
\end{eqnarray}

where 
\begin{eqnarray}
\mathbf{L}_{t} & =\mathbf{P}_{t \mid t} \mathbf{F}_{t+1} ^\mathrm{T} \mathbf{P}_{t+1 \mid t}^{-1}
\end{eqnarray}

The latter equation is the basic update equation in the RTS smoothing algorithm. This equation provides an estimate of $\mathbf{x}_{t}$ based on the filtered estimate $\hat{\mathbf{x}}_{t\mid t}$ corrected by the convolution of $\mathbf{L}_{t}$ with the error term $\mathbf{x}_{t+1 \mid T} - \hat{\mathbf{x}}_{t+1\mid t}$ that represents the difference between the smoothed estimate of $\mathbf{x}_{T}$  and the filtered estimate $\hat{\mathbf{x}}_{t+1\mid t}$.  The matrix $\mathbf{L}_{t}$ can be interpreted as a gain matrix that depends only on forward information. and can be computed in the forward pass.

As for the conditional variance, we have
\begin{eqnarray}
\hat{\mathbf{P}}_{t\mid T} & \triangleq & \operatorname{Var}[ \mathbf{x}_{t}  \mid \mathbf{z}_{1}, \ldots, \mathbf{z}_{T} ] \\
& = & \operatorname{Var}[ \mathbb{E}[ \mathbf{x}_{t}  \mid  \mathbf{x}_{t+1}, \mathbf{z}_{1}, \ldots, \mathbf{z}_{T} ]  \mid \mathbf{z}_{1}, \ldots, \mathbf{z}_{T} ]  \\
& & \qquad \quad  +
\mathbb{E}[ \operatorname{Var}[ \mathbf{x}_{t}  \mid  \mathbf{x}_{t+1}, \mathbf{z}_{1}, \ldots, \mathbf{z}_{T} ]  \mid \mathbf{z}_{1}, \ldots, \mathbf{z}_{T} ] \\
& = & \operatorname{Var}[ \mathbf{x}_{t \mid t}  +  \mathbf{L}_{t}(\mathbf{x}_{t+1 \mid T} - \hat{\mathbf{x}}_{t+1\mid t} )\mid \mathbf{z}_{1}, \ldots, \mathbf{z}_{T} ] \\
& & \qquad \quad + \mathbb{E}[  \mathbf{P}_{t \mid t} - \mathbf{L}_{t}  \mathbf{P}_{t+1 \mid t} \mathbf{L}_{t}^\mathrm{T} \mid \mathbf{z}_{1}, \ldots, \mathbf{z}_{T} ] \\
& = & \mathbf{L}_{t} \operatorname{Var}[ \mathbf{x}_{t \mid t}\mid \mathbf{z}_{1}, \ldots, \mathbf{z}_{T} ] \mathbf{L}_{t}^\mathrm{T} + \mathbf{P}_{t \mid t} - \mathbf{L}_{t}  \mathbf{P}_{t+1 \mid t} \mathbf{L}_{t}^\mathrm{T} \\
& = & \mathbf{L}_{t} \mathbf{P}_{t+1 \mid T} \mathbf{L}_{t}^\mathrm{T} + \mathbf{P}_{t \mid t} - \mathbf{L}_{t}  \mathbf{P}_{t+1 \mid t} \mathbf{L}_{t}^\mathrm{T} \\
& = & \mathbf{P}_{t \mid t} + \mathbf{L}_{t} (\mathbf{P}_{t+1 \mid T}-\mathbf{P}_{t+1 \mid t}) \mathbf{L}_{t}^\mathrm{T} 
\end{eqnarray}

We can summarize the RTS smoothing algorithm as follows:
\begin{eqnarray}
\hat{\mathbf{x}}_{t\mid T} & = & \hat{\mathbf{x}}_{t\mid t}  +  \mathbf{L}_{t}(\mathbf{x}_{t+1 \mid T} - \hat{\mathbf{x}}_{t+1\mid t} ) \\
\hat{\mathbf{P}}_{t\mid T} & = & \mathbf{P}_{t \mid t} + \mathbf{L}_{t} (\mathbf{P}_{t+1 \mid T}-\mathbf{P}_{t+1 \mid t}) \mathbf{L}_{t}^\mathrm{T} 
\end{eqnarray}

with an initial condition given by
\begin{eqnarray}
\hat{\mathbf{x}}_{T \mid T} & = & \hat{\mathbf{x}}_{T} \\
\hat{\mathbf{P}}_{T \mid T} & = & \hat{\mathbf{P}}_{T}
\end{eqnarray}
\end{proof}

\subsection{Proof of the inverse dynamics}\label{proof_inverse_dynamics}
\begin{proof}
We have 
\begin{eqnarray}
\mathbf{x}_{t}  & =&  \mathbf{F}_{t+1}^{-1} \left( \mathbf{x}_{t+1} - \mathbf{B}_{t+1} \mathbf{u}_{t+1} - w_{t+1} \right) \\
& =&  \widetilde{\mathbf{F}}_{t+1} \mathbf{x}_{t+1} + \mathbf{F}_{t+1}^{-1} \left(  (\mathbf{I} - \mathbf{F}_{t+1} \widetilde{\mathbf{F}}_{t+1}) \mathbf{x}_{t+1}  - \mathbf{B}_{t+1} \mathbf{u}_{t+1} - w_{t+1} \right) \\
& =&  \widetilde{\mathbf{F}}_{t+1} \mathbf{x}_{t+1} + \mathbf{F}_{t+1}^{-1} \left( \mathbf{Q}_{t+1} \mathbf{P}_{t+1}^\mathrm{-1}  \mathbf{x}_{t+1}  - \mathbf{B}_{t+1} \mathbf{u}_{t+1} - w_{t+1} \right) \\
& =&  \widetilde{\mathbf{F}}_{t+1} \mathbf{x}_{t+1} - \mathbf{F}_{t+1}^{-1} \mathbf{B}_{t+1} \mathbf{u}_{t+1} - \mathbf{F}_{t+1}^{-1} \left( w_{t+1}  - \mathbf{Q}_{t+1} \mathbf{P}_{t+1}^\mathrm{-1}  \mathbf{x}_{t+1}  \right)
\end{eqnarray}
which provides the first result \ref{inverse_dynamics}.

Using the forward dynamics (equation (\ref{forward_dynamics})) that relates $\mathbf{x}_{t+1}$ and ${w}_{t+1}$, we get that 
\begin{eqnarray}
\mathbb{E}[ \widetilde{w}_{t+1} \widetilde{w}_{t+1}^\mathrm{T} ] &=& \mathbf{F}_{t+1}^{-1} \mathbf{Q}_{t+1} (\mathbf{I} - \mathbf{P}_{t+1}^{-1} \mathbf{Q}_{t+1}) \mathbf{F}_{t+1}^{-\mathrm{T}},
\end{eqnarray}
which proves  \ref{inverse_dynamics_cov}.

The independence between $\widetilde{w}_{t+1}$ and the past information $\mathbf{x}_{t+1}, \ldots, \mathbf{x}_{T}$ is obvious for $\mathbf{x}_{t+2}, \ldots, \mathbf{x}_{T}$ and inferred from the fact that the two processes are Gaussian with zero correlation: $ \mathbb{E}[ \widetilde{w}_{t+1} \mathbf{x}_{t+1}] = -  \mathbf{F}_{t+1}^{-1} (\mathbf{Q}_{t+1} - \mathbf{Q}_{t+1}) = 0 $. 

Last but not least, using the independence, the forward Lyapunov equation is trivially derived, which concludes the proof.
\end{proof}

\subsection{Proof of modified Bryson–Frazier equations}
\label{proof_Bryson_Frazier}
\begin{proof}
Applying the information filter given in proposition \ref{information_filter} with the forward dynamics equation \ref{inverse_dynamics} gives:
\begin{eqnarray}
\widetilde{\mathbf{M}}_t & = & \widetilde{\mathbf{Q}}_{t}^{-1} \widetilde{\mathbf{F}}_{t} (\mathbf{\Lambda}_{t+1 \mid t+1} +\widetilde{\mathbf{F}}_{t}^\mathrm{T}  \widetilde{\mathbf{Q}}_{t}^{-1} \widetilde{\mathbf{F}}_{t})^{-1} \\
\mathbf{\Lambda}_{t \mid t+1} &=&  \widetilde{\mathbf{Q}}_t^{-1} - \widetilde{\mathbf{M}}_t \widetilde{\mathbf{F}}_{t}^\mathrm{T} \widetilde{\mathbf{Q}}_t^{-1} \\
\hat{\eta}_{t \mid t+1} &=& \widetilde{\mathbf{M}}_t  \hat{\eta}_{t+1 \mid t+1} + \mathbf{\Lambda}_{t \mid t+1} \widetilde{\mathbf{B}}_t  \mathbf{u}_t \\
\hat{\eta}_{t \mid t}&=&\hat{\eta}_{t \mid t+1}+\mathbf{H}_{t}^\mathrm{T}\mathbf{R}_{t}^{-1} \mathbf{z}_{t}\\
\mathbf{\Lambda}_{t \mid t} &=&  \mathbf{\Lambda}_{t \mid t+1} +\mathbf{H}_{t}^\mathrm{T} \mathbf{R}_{t}^{-1}\mathbf{H}_{t} 
\end{eqnarray}

We can notice that 
\begin{eqnarray}
\widetilde{\mathbf{Q}}_{t}^{-1} \widetilde{\mathbf{F}}_{t}  & = &  \mathbf{F}_{t}^{T}  \mathbf{Q}_{t}^{-1} \\
\widetilde{\mathbf{F}}_{t}^\mathrm{T}  \widetilde{\mathbf{Q}}_{t}^{-1} \widetilde{\mathbf{F}}_{t} &=&  \mathbf{Q}_{t}^{-1} - \mathbf{P}_{t}^{-1}  \\
\widetilde{\mathbf{F}}_{t}^\mathrm{T} \widetilde{\mathbf{Q}}_{t}^{-1}   & = &  \mathbf{Q}_{t}^{-1} \mathbf{F}_{t}  \\
\end{eqnarray}

Hence, if we rewrite everything in terms of the initial variables, we get
\begin{eqnarray}
\widetilde{\mathbf{M}}_t & = &  \mathbf{F}_{t}^{T}  \mathbf{Q}_{t}^{-1} (\mathbf{\Lambda}_{t+1 \mid t+1} +\ \mathbf{Q}_{t}^{-1} - \mathbf{P}_{t}^{-1})^{-1} \hspace{1.5cm} \\
\mathbf{\Lambda}_{t \mid t+1} &=&  \mathbf{F}_{t}^{T} (\mathbf{Q}_t - \mathbf{Q}_t \mathbf{P}_{t}^{-1} \mathbf{Q}_{t})^{-1}   \mathbf{F}_{t}  - \widetilde{\mathbf{M}}_t \mathbf{Q}_{t}^{-1} \mathbf{F}_{t} \\
\hat{\eta}_{t \mid t+1} &=& \widetilde{\mathbf{M}}_t  \hat{\eta}_{t+1 \mid t+1} - \mathbf{\Lambda}_{t \mid t+1}  \mathbf{F}_{t} \mathbf{B}_t  \mathbf{u}_t \\
\hat{\eta}_{t \mid t}&=&\hat{\eta}_{t \mid t+1}+\mathbf{H}_{t}^\mathrm{T}\mathbf{R}_{t}^{-1} \mathbf{z}_{t}\\
\mathbf{\Lambda}_{t \mid t} &=&  \mathbf{\Lambda}_{t \mid t+1} +\mathbf{H}_{t}^\mathrm{T} \mathbf{R}_{t}^{-1}\mathbf{H}_{t} 
\end{eqnarray}

which concludes the proof.
\end{proof}

\section{A few formula}
\subsection{Inversion matrix}
We state here without proof the various variant of the Woodbury identity as its proof is ubiquous in for instance \cite{Petersen_2012} or in \cite{wiki:Woodbury}

\begin{eqnarray}
\mkern-14mu ( \mathbf{A} + \mathbf{C} \mathbf{B} \mathbf{C}^\mathrm{T} ) ^{-1} &=&  \mathbf{A} ^{-1} -  \mathbf{A} ^{-1}\mathbf{C}( \mathbf{B}^{-1} + \mathbf{C}^\mathrm{T}    \mathbf{A} ^{-1} \mathbf{C})^{-1} \mathbf{C}^\mathrm{T}  \mathbf{A} ^{-1} \qquad \quad
\end{eqnarray}

If $ \mathbf{P},  \mathbf{R}$ are positive definite, we have
\begin{eqnarray}
( \mathbf{P}^{-1} + \mathbf{B}^\mathrm{T}  \mathbf{R}^{-1} \mathbf{B})^{-1}\mathbf{B}^\mathrm{T}  \mathbf{R}^{-1}  & =&  \mathbf{P} \mathbf{B}^\mathrm{T} (\mathbf{B}  \mathbf{P} \mathbf{B}^\mathrm{T} + \mathbf{R})^{-1} \qquad \quad
\end{eqnarray}

\subsection{Conditional Formula}\label{conditional_formula}
We have the following easy formulae for conditional expectation and variance (referred to as the tower equalities)
\begin{equation}
\mathbb{E}\left[X \mid Z \right] = \mathbb{E}\left[\mathbb{E}\left[ X \mid Y, Z \right] \mid Z \right] 
\end{equation}

\begin{equation}
\operatorname{Var}\left[X \mid Z \right] = \operatorname{Var}\left[\mathbb{E}\left[ X \mid Y, Z \right] \mid Z \right] + \mathbb{E}\left[\operatorname{Var}\left[X \mid Y, Z \right] \mid Z \right]
\end{equation}

\begin{proof}
The conditional expectation $\mathbb E[X \mid W = w, Y = y]$ is a function of $w$ and $y$ and can be denoted by $g(w,y)$. 
We can now calculate as follows:

\begin{align*}
\mathbb E [ \mathbb E [X \mid W,Y] \mid Y= y ] & = \mathbb E [g(W,Y) \mid Y = y] \\
&= \sum_{w,y^\prime} g\left(w,y^\prime \right) \Pr \left[\left.W = w, Y = y^\prime \right\vert Y = y\right] \\
&= \sum_w g(w,y) \Pr [W = w \mid Y = y] \\
&= \sum_w \mathbb E [X \mid W=w,Y=y] \Pr [W = w \mid Y = y] \\
&= \sum_w \left[ \sum_x x \Pr [X=x \mid W=w,Y=y] \right] \\
& \qquad \qquad \qquad \Pr [W = w \mid Y = y] \\
&= \sum_{w,x} x \Pr[X=x, W=w \mid Y = y] \\
&= \sum_x x \Pr[X = x \mid Y = y] \\
&= \mathbb E[X \mid Y=y]
\end{align*}

where in our series of equation, the third equality follows from the fact that the conditional probability given by 
$\Pr \left[ \left. W = w, Y = y^\prime  \right\vert Y = y \right]$ is $\Pr[W = w \mid Y= y]$ for $y^\prime = y$ 
and $0$ otherwise and the sixth equality comes from Bayes' rule.

The second equality is trivial and is a consequence of the first and the law of total variance.
\end{proof}

\section{EM Algorithm convergence proofs}
\subsection{First Proof}\label{EMProof1}
\begin{proof}
Expectation-maximization improves the expected complete log likelihood: $Q(\boldsymbol\theta|\boldsymbol\theta^{(t)})$ rather than the complete log likehood: $\log p(\mathbf{X}|\boldsymbol\theta)$. We have for any $\mathbf{Z}$ with non-zero probability $p(\mathbf{Z}|\mathbf{X},\boldsymbol\theta)$ that the log likehood conditional to the parameters $\boldsymbol\theta$ can be split into two parts:

\begin{equation*}
\log p(\mathbf{X}|\boldsymbol\theta) = \log p(\mathbf{X},\mathbf{Z}|\boldsymbol\theta) - \log p(\mathbf{Z}|\mathbf{X},\boldsymbol\theta) \,.
\end{equation*}

Taking the expectation over possible values for the latent variables $\mathbf{Z}$ under the current parameter estimate $\theta^{(t)}$, 
multiplying both sides by $p(\mathbf{Z}|\mathbf{X},\boldsymbol\theta^{(t)})$ and summing (or integrating) over $\mathbf{Z}$, we get:

\begin{equation}\label{EM_eq1}
\begin{aligned}
\log p(\mathbf{X}|\boldsymbol\theta) &
= \sum_{\mathbf{Z}} p(\mathbf{Z}|\mathbf{X},\boldsymbol\theta^{(t)}) \log p(\mathbf{X},\mathbf{Z}|\boldsymbol\theta)
- \sum_{\mathbf{Z}} p(\mathbf{Z}|\mathbf{X},\boldsymbol\theta^{(t)}) \log p(\mathbf{Z}|\mathbf{X},\boldsymbol\theta) \\
& = Q(\boldsymbol\theta|\boldsymbol\theta^{(t)}) + H(\boldsymbol\theta|\boldsymbol\theta^{(t)}) \,,
\end{aligned}
\end{equation}

where $H(\boldsymbol\theta|\boldsymbol\theta^{(t)})$ denotes the negated sum it is replacing. As this last equation holds for any value of $\boldsymbol\theta$ including $\boldsymbol\theta = \boldsymbol\theta^{(t)}$, we have:
\begin{equation}\label{EM_eq2}
\log p(\mathbf{X}|\boldsymbol\theta^{(t)})
= Q(\boldsymbol\theta^{(t)}|\boldsymbol\theta^{(t)}) + H(\boldsymbol\theta^{(t)}|\boldsymbol\theta^{(t)}) \,,
\end{equation}

We can now subtract \ref{EM_eq2} to \ref{EM_eq1} to get:
\begin{equation*}
\log p(\mathbf{X}|\boldsymbol\theta) - \log p(\mathbf{X}|\boldsymbol\theta^{(t)})
= Q(\boldsymbol\theta|\boldsymbol\theta^{(t)}) - Q(\boldsymbol\theta^{(t)}|\boldsymbol\theta^{(t)})
 + H(\boldsymbol\theta|\boldsymbol\theta^{(t)}) - H(\boldsymbol\theta^{(t)}|\boldsymbol\theta^{(t)}) \,,
\end{equation*}

It is easy to conclude using Gibbs' inequality that tells us that $H(\boldsymbol\theta|\boldsymbol\theta^{(t)}) \ge H(\boldsymbol\theta^{(t)}|\boldsymbol\theta^{(t)})$ and get:
\begin{equation*}
\log p(\mathbf{X}|\boldsymbol\theta) - \log p(\mathbf{X}|\boldsymbol\theta^{(t)})
\ge Q(\boldsymbol\theta|\boldsymbol\theta^{(t)}) - Q(\boldsymbol\theta^{(t)}|\boldsymbol\theta^{(t)}) \,.
\end{equation*}
which states that the marginal likelihood at each step is non-decreasing, and hence concludes the proof.
\end{proof}

\subsection{Second Proof}\label{EMProof2}
\begin{proof}
Another proof relies on the fact that the EM algorithm can be viewed as a coordinate ascent method, which is proved to converges monotonically to a local minimum of the function we maximize (see for instance \cite{Hastie_2009} or \cite{Neal_1999}). This works as follows. Let us consider the function:
\begin{equation*}
F(q,\theta) := \operatorname{E}_q [ \log L (\theta ; x,Z) ] + H(q), 
\end{equation*}
where $q$ is an arbitrary probability distribution over latent data $\mathbf{Z}$ and $H(q)$ is the Entropy of the distribution $q$. This can also be written as:

\begin{equation*}
F(q,\theta) = -\mathrm{Div}_{\mathrm{KL}}\big(q \big\| p_{Z|X}(\cdot|x;\theta ) \big) + \log L(\theta;x), 
\end{equation*}
where  $p_{Z|X}(\cdot|x;\theta )$ is the conditional distribution of the latent data given the observed data $\mathbf{X}$ and $\mathrm{Div}_{\mathrm{KL}}$ is the Kullback–Leibler divergence.

In the EM algorithm, the expectation step is to choose $q$ to maximize $F$:
\begin{equation*}
q^{(t)} = \operatorname{arg\,max}_q \ F(q,\theta^{(t)}) 
\end{equation*}
The maximization step is to choose the parameter $\theta$ to maximize $F$:
\begin{equation*}
\theta^{(t+1)} = \operatorname{arg\,max}_\theta \ F(q^{(t)},\theta) 
\end{equation*}

which concludes the proof by showing that the EM algorithm is a coordinate ascent method.
\end{proof}

\newpage
\bibliography{mybib}

\end{document}